\documentclass[a4paper,10pt]{article}

\usepackage[utf8]{inputenc}
\usepackage[T1]{fontenc}

\usepackage{a4wide}

\usepackage{amsmath,amssymb}
\usepackage{amsthm}
\usepackage{enumerate} 

\usepackage{authblk}

\usepackage[svgnames]{xcolor}
\definecolor{highlightNEW}{named}{black}

\usepackage{hyperref}
\hypersetup{
    colorlinks=true,
    linkcolor=Navy,
    citecolor=Navy,
    urlcolor=Navy
}

\usepackage{graphicx}
\graphicspath{{fig/}}
\usepackage{epstopdf}
\usepackage{subcaption}

\usepackage{booktabs}
\usepackage{multirow}

\usepackage[absolute,overlay,showboxes]{textpos}
\TPMargin{2mm}
\textblockrulecolour{Navy}

\usepackage[round]{natbib}
\bibpunct[, ]{(}{)}{;}{a}{}{,}

\newtheorem{theorem}{Theorem}[section] 
\newtheorem{corollary}[theorem]{Corollary} 
 
\newtheorem{lemma}[theorem]{Lemma} 
 
\newtheorem{remark}[theorem]{Remark}

\newcommand{\doi}[1]{DOI~\href{\detokenize{http://dx.doi.org/#1}}{\detokenize{#1}}}

\newdimen\CdotAxis
\newcommand*{\CdotAux}[3]{%
  {%
    \settoheight\CdotAxis{$#2\vcenter{}$}%
    \sbox0{%
      \raisebox\CdotAxis{%
        \scalebox{#1}{%
          \raisebox{-\CdotAxis}{%
            $\mathsurround=0pt #2#3$%
          }%
        }%
      }%
    }%
    \dp0=0pt %
    \sbox2{$#2\bullet$}%
    \ifdim\ht2<\ht0 %
      \ht0=\ht2 %
    \fi
    \sbox2{$\mathsurround=0pt #2#3$}%
    \hbox to \wd2{\hss\usebox{0}\hss}%
  }%
}
\makeatletter
\def\mathcolor#1#{\@mathcolor{#1}}
\def\@mathcolor#1#2#3{%
  \protect\leavevmode
  \begingroup
    \color#1{#2}#3%
  \endgroup
}
\makeatother

\let\oldalpha\alpha
\renewcommand{\alpha}{\mathcolor{highlightNEW}{\oldalpha}}
\newcommand{\ccode}[2]{\par
        \vspace*{8pt}
        {{\leftskip18pt\rightskip\leftskip
        \noindent{\it #1}\/: #2\par}}\par}
\newcommand{\keywords}[1]{\ccode{Keywords}{#1}}
\newcommand{\email}[1]{\href{mailto:#1}{#1}}

\def\received#1{Received~#1\par}
\def\revised#1{Revised~#1\par}
\def\accepted#1{Accepted~#1\par}

\def\foliofont{\fontsize{8}{10}\selectfont}

\newcommand{\jpTitle}{Decomposition formula for jump diffusion models}
\newcommand{\jpAuthors}{R. Merino, J. Posp\'{\i}\v{s}il, T. Sobotka, and J. Vives}
\newcommand{\jpKeywords}{option pricing; stochastic volatility models; jump diffusion models; implied volatility}
\newcommand{\jpMSC}{60G51; 91G20; 91G60}
\newcommand{\jpJEL}{G12; C58; C63}
\newcommand{\jpDateReceived}{2 March 2018}

\newcommand{\jpDate}{}%\today
\newcommand{\jpPreprint}{Preprint }

\author[1,3]{Ra\'{u}l Merino}
\author[2]{Jan Posp\'{\i}\v{s}il\thanks{Corresponding author, \email{honik@kma.zcu.cz}}}
\author[2]{Tom\'{a}\v{s} Sobotka}
\author[1]{Josep Vives}
\affil[1]{Facultat de Matem\`{a}tiques i Inform\`{a}tica, Universitat de Barcelona, \authorcr Gran Via 585, 08007 Barcelona, Spain,\vspace*{3pt}}
\affil[2]{NTIS - New Technologies for the Information Society, Faculty of Applied Sciences, \authorcr University of West Bohemia, Univerzitn\'{\i} 8, 301 00 Plze\v{n}, Czech Republic,\vspace*{3pt}}
\affil[3]{VidaCaixa S.A., Investment Risk Management Department, \authorcr C/Juan Gris, 2-8, 08014 Barcelona, Spain.}

\ifpdf
\hypersetup{
  pdftitle={\jpTitle},
  pdfauthor={\jpAuthors},
  pdfkeywords={\jpKeywords},
  pdfinfo={
      MSC={\jpMSC},
      JEL={\jpJEL}
  }
}
\fi

\title{\textcolor{Navy}{\textsc{\jpTitle}}}
\date{\jpDate}

\begin{document}

\maketitle
\begin{textblock}{6}(1.25,1)
{\foliofont\noindent \jpPreprint of an article published in International Journal of Theoretical and Applied Finance \\ % Int. J. Theor. Appl. Finance \\ 
Vol. 21, No. 8 (2018) 1850052, \doi{10.1142/S0219024918500528} \\
\textcopyright World Scientific Publishing Company, \href{https://www.worldscientific.com/ijtaf}{www.worldscientific.com/ijtaf}
}
\end{textblock}

\begin{center}
\received{\jpDateReceived}
\revised{28 September 2018}
\accepted{1 October 2018}
\end{center}

\begin{abstract}
% ------------------------------------------------------------------------------ begin: abstract.tex
In this paper we derive a generic decomposition of the option pricing formula for models with finite activity jumps in the underlying asset price process (SVJ models). This is an extension of the well-known result by \cite{Alos12} for \cite{Heston93} SV model. Moreover, explicit approximation formulas for option prices are introduced for a popular class of SVJ models - models utilizing a variance process postulated by \cite{Heston93}. In particular, we inspect in detail the approximation formula for the \cite{Bates96} model with log-normal jump sizes and we provide a numerical comparison with the industry standard - Fourier transform pricing methodology. For this model, we also reformulate the approximation formula in terms of implied volatilities. The main advantages of the introduced pricing approximations are twofold. Firstly, we are able to significantly improve computation efficiency (while preserving reasonable approximation errors) and secondly, the formula can provide an intuition on the volatility smile behaviour under a specific SVJ model.

% ------------------------------------------------------------------------------ end: abstract.tex
\end{abstract}

\keywords{\jpKeywords}
\ccode{MSC classification}{\jpMSC}
\ccode{JEL classification}{\jpJEL}

\allowdisplaybreaks
% ------------------------------------------------------------------------------ begin: introduction.tex
\section{Introduction}\label{sec:introduction}

The main problem of the Black-Scholes option pricing model is the assumption of constant volatility for the underlying stock price process. In practice, this model is used as a marking model to quote implied volatilities instead of traded option prices. Contrary to the model assumptions, the implied volatilities observed in the vanilla option markets are not flat - they typically exhibit a non-zero skew and a convex smile-like shape in the moneyness dimension.  To correctly capture the shape of implied volatility surfaces, various stochastic volatility (SV) models were developed. These models assume that not only the spot prices are stochastic, but also their volatility is driven by a suitable stochastic process. Another way how to deal with drawbacks of the Black-Scholes model is to add a jump term to the stock price process. This results into the jump diffusion setting, which was originally studied by \cite{Merton76}. In this article, we build an option price approximation framework for a popular class of financial models that utilize both of the aforementioned ideas. Hence, the main objects of our study are stochastic volatility jump diffusion (SVJ) models.

The first SVJ model is credited to \cite{Bates96} who incorporated a stochastic variance process postulated by \cite{Heston93} alongside \cite{Merton76} - style jumps. The variance of stock prices follows a CIR process \citep{CIR85} and the stock prices themselves are assumed to be of a jump diffusion type with log-normal jump sizes. In particular, this model should improve the market fit for short-term maturity options, while the original \cite{Heston93} approach would often need unrealistically high volatility of variance parameter to fit reasonably well the short-term smile \citep{Bayer15,MrazekPospisilSobotka16ejor}. An SVJ model with a non-constant interest rate was introduced by \cite{Scott97}. Several other authors studied SVJ models that have a different distribution for jump sizes, e.g. \cite{YanHanson06} utilized log-uniform jump amplitudes.

Naturally, one can extend SVJ models by adding jumps into the variance process (e.g. a model introduced by \cite{Duffie00}). However, based on several empirical studies, these models tend to overfit market prices and despite having more parameters than the original \cite{Bates96} model they might not provide a better calibration errors (see e.g. \cite{Gatheral06}). Another way to improve standard SV models might be to introduce time-dependent model parameters. The \cite{Heston93} model with time-dependent parameters was studied by \cite{MikhailovNogel03} for piece-wise constant parameters, by
\cite{Elices08} for a linear dependence and a more general modification was introduced by \cite{Benhamou10}. These approaches involve several additional parameters and might also suffer from overfitting. Moreover, \cite{Bayer15} mentioned that these models do not fully comply with properties of observable market data - a general overall shape of the volatility surface typically does not change in time and hence the option prices should be derived using a time-homogeneous stochastic process. 

The valuation of derivatives under these more complex models is, of course, a more elaborate task compared to the standard Black-Scholes model. Many authors have introduced semi-closed form formulas using various transformation techniques of the pricing partial (integro) differential equations, to name a few: \cite{Heston93}, \cite{Bates96}, \cite{Scott97}, \cite{Lewis00}, \cite{Albrecher07}, \cite{BaustianMrazekPospisilSobotka17asmb} and many others. Although transform pricing methods are typically efficient tools to evaluate non-path dependent derivatives, they do not provide any intuition on the smile behavior. Moreover, calibration routines utilizing these methods lead typically to non-convex optimization problems (see e.g. \cite{MrazekPospisilSobotka16ejor}). 

Other authors considered approximation techniques that were pioneered by \cite{HullWhite87}. In the last years, the \cite{HullWhite87} pricing formula was reinvented using techniques of the Malliavin calculus, because a future average volatility that is used in the formula is a non adapted stochastic process. In \cite{Alos06}, \cite{Alos07} and \cite{Alos08}, a general jump diffusion model with no prescribed volatility process is analyzed. There have been several extensions thereof, e.g. by assuming Lévy processes in \cite{JafariVives13}, see also the survey in \cite{Vives16}.
 
In \cite{Alos12}, a new approach of dealing with the Hull and White formula and the Heston model has been proposed. The main idea of this approach is to use an adapted projection for the future volatility. The formula provides a valuable intuition on the behavior of smiles and term structures under the Heston model. This is not a purely theoretical result - it can significantly fasten/improve the calibration process by providing a good initial guess by analytical calibration or by specifying a region where calibrated parameters should lie in as it is done in \cite{Alos15}. In \cite{MerinoVives15}, the idea of \cite{Alos12} has been used to find a general decomposition formula for any stochastic volatility process satisfying basic integrability conditions.

In the present paper, we apply the same set of ideas and we extend them to the domain of SVJ models with finite activity jumps. This should serve not only to find a more efficient way to price vanilla options compared to transform pricing methods (see Section \ref{sec:heston}), but as a side product we provide a similar intuition of the smile behavior for the studied SVJ model.

In particular, we start by finding a generic decomposition formula for a vanilla call option price and an approximation for both the price and implied volatility under a specific SVJ model. Explicit pricing formulas are provided for one of the most popular SVJ models - \cite{Heston93} type models with compound Poisson process in the stock price evolution. To assess the accuracy and efficiency of the newly derived solution, we perform a numerical comparison for the \cite{Bates96} model (i.e. log-normal jump sizes) alongside its Fourier transform pricing formula introduced by \cite{BaustianMrazekPospisilSobotka17asmb}.

The structure of the paper is as follows.
In Section \ref{sec:preliminaries}, we give basic preliminaries and our notation related to SVJ models. This notation will be used throughout the paper without being repeated in particular theorems, unless we find useful to do so in order to guide the reader through the results. 
In Sections \ref{sec:Generic_decomposition} and \ref{sec:model}, we derive decomposition formulas for SV and SVJ models, respectively, generalizing the decomposition formula obtained by \cite{Alos12}. Newly obtained decomposition is rather versatile since it does not need to specify the underlying volatility process. 
Particular approximation formulas for several SVJ models are presented in Section \ref{sec:heston} alongside the numerical comparison for the \cite{Bates96} model. The decomposition result in terms of implied volatilities is introduced in Section \ref{sec:implied}.
A discussion of the results is provided in Section \ref{sec:conclusion} and technical error estimates are presented in %Appendix 
\ref{sec:appendix}. 

% ------------------------------------------------------------------------------ end: introduction.tex
% ------------------------------------------------------------------------------ begin: preliminaries.tex
\section{Preliminaries and notation}\label{sec:preliminaries}

Let $S=\{S_{t}, t\in [0,T]\}$ be a strictly positive price process under a market chosen risk neutral probability that follows the model:  

\begin{eqnarray}\label{model}
 dS_{t}= rS_tdt + \sigma_{t}S_{t}  \left(\rho dW_{t}  + \sqrt{1-\rho^{2}}d\tilde{W}_{t} \right) + S_{t-}dZ_{t},
\end{eqnarray}
where $S_{0}$ is the current price, $W$ and $\tilde{W}$ are independent Brownian motions, $r$ is the interest rate, $\rho \in (-1,1)$ is the correlation between the two Brownian motions and 

$$Z_t=\int_0^t \int_{{\mathbb R}} (e^y-1){\tilde N}(ds, dy)$$
where $N$ and ${\tilde N}$ denote the Poisson measure and the compensated Poisson measure, respectively. We can associate to measure $N$ a compound Poisson process $J$, independent of $W$ and $\tilde{W}$, with intensity $\lambda \geq 0$ and jump amplitudes given by random variables $Y_i$, independent copies of a random variable $Y$ with law given by $Q$.  Recall that this compound Poisson process can be written as 

$$J_t:=\int_0^t \int_{\mathbb R} y N(ds,dy)=\sum_{i=1}^{n_t} Y_i,$$ 
where $n_t$ is a $\lambda-$ Poisson process. Denote by $k:={\mathbb E}_{Q}(e^Y-1).$

Without any loss of generality, it will be convenient in the following sections, to use as underlying process, the log-price process 
$X_{t} = \log S_{t}, t\in[0,T]$, that satisfies

\begin{eqnarray}\label{log-model}
dX_{t}= \left(r - \lambda k-\frac{1}{2}\sigma^{2}_{t}\right)dt + \sigma_{t} \left(\rho dW_{t}  + \sqrt{1-\rho^{2}}d\tilde{W}_{t} \right) + dJ_{t}.
\end{eqnarray}
We introduce also the corresponding continuous process,
\begin{eqnarray}\label{log-model-continious}
d\tilde{X}_{t}= \left(r - \lambda k -\frac{1}{2}\sigma^{2}_{t}\right)dt + \sigma_{t} \left(\rho dW_{t}  + \sqrt{1-\rho^{2}}d\tilde{W}_{t} \right).
\end{eqnarray}

The volatility process $\sigma$ is a square-integrable process assumed to be adapted to the filtration generated by $W$ and $J$ and its trajectories are assumed to be a.s. square integrable, c\`adl\`ag and strictly positive a.e. 

\begin{remark}
Observe that this is a very general stochastic volatility model. We can consider the following particular cases:
\begin{itemize}

\item If $\sigma$ is constant and we have finite activity jumps, we have a generic jump-diffusion model as for example the Merton model. In the particular case of 
$\sigma=0$ we have an exponential L\'evy model.

\item If we assume no jumps, that is $\lambda=0$, we have a generic stochastic volatility diffusion model. This is the case treated in \cite{MerinoVives15}.

\item If in addition $\rho=0$ we have a generalization of different non correlated stochastic volatility diffusion models as \cite{HullWhite87}, \cite{Scott87}, \cite{SteinStein91} or \cite{BallRoma94}. 

\item If we assume no correlation but presence of jumps we cover for example the Heston-Kou model (e.g. see \cite{GulisashviliVives12}), or any uncorrelated model with the addition of finite activity Lévy jumps on the price process.  

\item Finally, if we have no jumps and $\sigma$ is constant, we have the classical Osborne-Samuelson-Black-Scholes model.
\end{itemize}
\end{remark}

The following notation will be used throughout the paper:

\begin{itemize}

\item We denote by ${\mathcal{F}}^W$, ${\mathcal{F}}^{\tilde{W}}$ and ${\mathcal{F}}^N$ the filtrations generated by the independent processes $W$, $\tilde{W}$ and $J$ respectively. 
Moreover, we define $\mathcal{F}:= {\mathcal{F}}^{W} \vee {\mathcal{F}}^{\tilde{W}}\vee {\mathcal{F}}^{N}$. 

\item We will denote by $BS(t,x,y)$ the price of a plain vanilla European call option under the classical Black-Scholes model with constant volatility $y$, current log stock price $x$, time to maturity $\tau=T-t$, strike price $K$ and interest rate $r$. In this case, 
\begin{eqnarray}
\nonumber BS\left(t,x, y\right)= e^{x} \Phi(d_{+}) - K e^{-r\tau} \Phi(d_{-}),
\end{eqnarray}
where $\Phi(\cdot)$ denotes the cumulative distribution function of the standard normal law and 
\begin{eqnarray}
\nonumber d_{\pm} = \frac{x-\ln K + (r \pm \frac{y^{2}}{2})\tau}{y\sqrt{\tau}}.
\end{eqnarray}

\item In our setting, the call option price is given by 

$$V_{t}=e^{-r\tau}{\mathbb E}_t [(e^{X_{T}}-K)^+].$$

\item  Recall that from the Feynman-Kac formula for the model \eqref{log-model-continious}, the operator 
\begin{eqnarray}{\label{FK}}
\mathcal{L}_{\sigma}:= {\partial}_t + \frac{1}{2}\sigma_{t}^{2} {\partial^{2}_{x}} 
					+ \left(r -\lambda k-\frac{1}{2}\sigma^{2}_{t}\right) {\partial}_{x}-r
\end{eqnarray}
satisfies ${\mathcal L}_{\sigma}BS(t,\tilde{X}_{t},\sigma_{t})=0$.

\item We define the operators $\Lambda:=\partial_x$, $\Gamma:=\left(\partial^{2}_{x}-\partial_{x}\right)$ and $\Gamma^{2}=\Gamma\circ \Gamma.$ In particular, for the Black-Scholes formula we obtain: 
\begin{eqnarray*}
\Gamma BS(t,x,y)&:=& \frac{e^{x}}{y\sqrt{2\pi \tau}}\exp\left(-\frac{d_+^2(y)}{2}\right),\\
\Lambda \Gamma BS(t,x,y)&:=& \frac{e^{x}}{y \sqrt{2\pi \tau}}\exp\left(-\frac{d_+^2(y)}{2}\right)\left(1-\frac{d_+(y)}{y \sqrt{\tau}}\right),\\
\Gamma^2 BS(t,x,y)&:=& \frac{e^{x}}{y \sqrt{2\pi \tau}}\exp \left(-\frac{d_+^2(y)}{2}\right)\frac{d_+^2(y)-y d_+(y) \sqrt{\tau}-1}{y^2\tau}.
\end{eqnarray*}

\item We define $p_{n}(\lambda T)$ as the Poisson probability mass function with intensity $\lambda T$. I.e. $p_n$ takes the following form:
$$p_n(\lambda T):=\frac{e^{-\lambda T} (\lambda T)^n}{n!}.$$

\end{itemize}

% ------------------------------------------------------------------------------ end: preliminaries.tex
% ------------------------------------------------------------------------------ begin: Generic_decomposition.tex
\section{A generic SV decomposition formula}\label{sec:Generic_decomposition}

In this section, following the ideas of \cite{Alos12}, see also \cite{MerinoVives15}, we extend the decomposition formula to a generic stochastic volatility model. We recall that the formula is valid without having to specify the underlying volatility process explicitly, which enables us to obtain a very flexible decomposition formula. The formula proved in \cite{Alos12} is the particular case of the Heston model. 

It is well known that if the stochastic volatility process is independent of the price process, then the pricing formula of a plain vanilla European call is given by 

$$V_{t}={\mathbb E}_t [BS(t,S_{t},{\bar \sigma}_{t})]$$  
where ${\bar \sigma}^2_{t}$ is the so called average future variance and it is defined by 

\begin{eqnarray}
\nonumber {\bar\sigma}^2_{t}:=\frac{1}{T-t} \int^{T}_{t}\sigma^{2}_{s}ds.
\end{eqnarray}
Naturally, ${\bar\sigma}_{t}$ is called the average future volatility, see \cite{Fouque00}, page 51.   

The idea used in \cite{Alos12} consists of using an adapted projection of the average future variance 
\begin{eqnarray}
\nonumber v^2_{t}:={\mathbb E}_t({\bar\sigma}^2_{t})=\frac{1}{T-t} \int^{T}_{t}{\mathbb E}_t[\sigma^{2}_{s}]ds
\end{eqnarray}
to obtain a decomposition of $V_{t}$ in terms of $v_{t}.$ This idea switches an anticipative problem related with the anticipative process ${\bar\sigma}_{t}$ into a 
non-anticipative one related to the adapted process $v_{t}$. 

We define
\begin{equation}
M_{t}=\int^{T}_{0}\mathbb{E}_{t}\left[\sigma^{2}_{s}\right]ds,\label{def:M_t}
\end{equation}
and hence
$$dv^{2}_{t}= \frac{1}{T-t}\left[dM_{t}+\left(v^{2}_{t}-\sigma^{2}_{t}\right)dt\right].$$

Recall that $M$ is a martingale with respect the filtration generated by $W$ and $J$.

The following processes will play an important role in a generic decomposition formula that will be introduced in this section. Let

\begin{equation}
R_{t}=\frac{1}{8}\mathbb{E}_{t}\left[\int^{T}_{t}d[M,M]_{u}\right]\label{def:R_t}
\end{equation}
and 

\begin{equation}
U_{t}=\frac{\rho}{2}\mathbb{E}_{t}\left[\int^{T}_{t}\sigma_{u}d[W,M]_{u}\right],\label{def:U_t}
\end{equation}
where $[\cdot, \cdot]$ denotes the quadratic covariation process. 

Now we prove a generic version of Theorem 2.2 in \cite{Alos12} which will be useful for our problem.\par

\begin{theorem}[Generic decomposition formula]
\label{Generic Deco}
Let $B_{t}$ be a continuous semimartingale with respect to the filtration $\mathcal{F}_t$, let $A(t,x,y)$ be a $C^{1,2,2} ([0,T]\times [0,\infty)\times [0,\infty))$ function and let $v^2_t, M_t$ be defined as above. Then we are able to formulate the expectation of $e^{-rT}A(T,\tilde{X}_{T},v^{2}_{T})B_{T}$
in the following way:
\begin{eqnarray*}
&&{\mathbb{E}}\left[e^{-rT}A(T,\tilde{X}_{T},v^{2}_{T})B_{T}\right]=A(0,\tilde{X}_{0},v^{2}_{0})B_{0}\\
&+&\mathbb{E}\left[\int^{T}_{0}e^{-ru}\partial_{y}A(u,\tilde{X}_{u},v^{2}_{u})B_{u}\frac{1}{T-u}\left(v^{2}_{u}-\sigma^{2}_{u}\right)du\right]\\
&+&\mathbb{E}\left[\int^{T}_{0}e^{-ru}A(u,\tilde{X}_{u},v^{2}_{u})dB_{u}\right]\\ 
&+&\frac{1}{2}\mathbb{E}\left[\int^{T}_{0}e^{-ru}\left(\partial^{2}_{x}-\partial_{x}\right)A(u,\tilde{X}_{u},v^{2}_{u})B_{u}\left(\sigma^{2}_{u}- v^{2}_{u}\right)du\right]\\ 
&+&\frac{1}{2}\mathbb{E}\left[\int^{T}_{0}e^{-ru}\partial^{2}_{y}A(u,\tilde{X}_{u},v^{2}_{u})B_{u}\frac{1}{(T-u)^{2}}d[M,M]_{u}\right]\\ 
&+&\rho \mathbb{E}\left[\int^{T}_{0}e^{-ru}\partial^{2}_{x,y}A(u,\tilde{X}_{u},v^{2}_{u})B_{u}\frac{\sigma_{u}}{T-u}d[W,M]_{u}\right]\\ 
&+&\sqrt{1-\rho^{2}} \mathbb{E}\left[\int^{T}_{0}e^{-ru}\partial^{2}_{x,y}A(u,\tilde{X}_{u},v^{2}_{u})B_{u}\frac{\sigma_{u}}{T-u}d[\tilde{W},M]_{u}\right]\\ 
&+&\rho \mathbb{E}\left[\int^{T}_{0}e^{-ru}\partial_{x}A(u,\tilde{X}_{u},v^{2}_{u})\sigma_{u} d[W,B]_{u}\right]\\ 
&+&\sqrt{1-\rho^{2}} \mathbb{E}\left[\int^{T}_{0}e^{-ru}\partial_{x}A(u,\tilde{X}_{u},v^{2}_{u})\sigma_{u} d[\tilde{W},B]_{u}\right]\\ 
&+&\mathbb{E}\left[\int^{T}_{0}e^{-ru}\partial_{y}A(u,\tilde{X}_{u},v^{2}_{u})\frac{1}{T-u}d[M,B]_{u}\right].
\end{eqnarray*}
\end{theorem}
\begin{proof}
Applying the Itô formula to the process $e^{-rt}A(t,{\tilde X}_{t},v^{2}_{t})B_{t}$ we obtain:
\begin{eqnarray*}
&&e^{-rT}A(T,\tilde{X}_{T},v^{2}_{T})B_{T}=A(0,\tilde{X}_{0},v^{2}_{0})B_{0}\\
&-&r\int^{T}_{0}e^{-ru}A(u,\tilde{X}_{u},v^{2}_{u})B_{u}du \\ 
&+&\int^{T}_{0}e^{-ru}\partial_{t}A(u,\tilde{X}_{u},v^{2}_{u})B_{u}du\\ 
&+&\int^{T}_{0}e^{-ru}\partial_{x}A(u,\tilde{X}_{u},v^{2}_{u})B_{u}d\tilde{X}_{u}\\
&+&\int^{T}_{0}e^{-ru}\partial_{y}A(u,\tilde{X}_{u},v^{2}_{u})B_{u}dv^{2}_{u}\\
&+&\int^{T}_{0}e^{-ru}A(u,\tilde{X}_{u},v^{2}_{u})dB_{u}\\ 
&+&\frac{1}{2}\int^{T}_{0}e^{-ru}\partial^{2}_{x}A(u,\tilde{X}_{u},v^{2}_{u})B_{u}d[\tilde{X},\tilde{X}]_{u}\\ 
&+&\frac{1}{2}\int^{T}_{0}e^{-ru}\partial^{2}_{y}A(u,\tilde{X}_{u},v^{2}_{u})B_{u}d[v^{2},v^{2}]_{u}\\ 
&+&\int^{T}_{0}e^{-ru}\partial^{2}_{x,y}A(u,\tilde{X}_{u},v^{2}_{u})B_{u}d[\tilde{X},v^{2}]_{u}\\ 
&+&\int^{T}_{0}e^{-ru}\partial_{x}A(u,\tilde{X}_{u},v^{2}_{u})d[\tilde{X},B]_{u}\\ 
&+&\int^{T}_{0}e^{-ru}\partial_{y}A(u,\tilde{X}_{u},v^{2}_{u})d[v^{2},B]_{u}.
\end{eqnarray*}
In the next step we apply the Feynman-Kac operator with volatility $v_{t}$, alongside the definition of $M_t$. After algebraic operations, we retrieve
\begin{eqnarray*}
&&e^{-rT}A(T,\tilde{X}_{t},v^{2}_{T})B_{T}=A(0,\tilde{X}_{0},v^{2}_{0})B_{0}\\
&+&\frac{1}{2}\int^{T}_{0}e^{-ru}\partial_{x}A(u,\tilde{X}_{u},v^{2}_{u})B_{u}(v^{2}_{u}-\sigma^{2}_{u})du\\
&+&\int^{T}_{0}e^{-ru}\partial_{x}A(u,\tilde{X}_{u},v^{2}_{u})B_{u} \sigma_{u}(\rho dW_{u}+\sqrt{1-\rho^2}d\tilde{W}_{u})\\
&+&\int^{T}_{0}e^{-ru}\partial_{y}A(u,\tilde{X}_{u},v^{2}_{u})B_{u}\frac{1}{T-u}dM_{u}\\
&+&\int^{T}_{0}e^{-ru}\partial_{y}A(u,\tilde{X}_{u},v^{2}_{u})B_{u}\frac{1}{T-u}\left(v^{2}_{u}-\sigma^{2}_{u}\right)du\\
&+&\int^{T}_{0}e^{-ru}A(u,\tilde{X}_{u},v^{2}_{u})dB_{u}\\ 
&+&\frac{1}{2}\int^{T}_{0}e^{-ru}\partial^{2}_{x}A(u,\tilde{X}_{u},v^{2}_{u})B_{u}\left(\sigma^{2}_{u}- v^{2}_{u}\right)du\\ 
&+&\frac{1}{2}\int^{T}_{0}e^{-ru}\partial^{2}_{y}A(u,\tilde{X}_{u},v^{2}_{u})B_{u}\frac{1}{(T-u)^{2}}d[M,M]_{u}\\ 
&+&\rho\int^{T}_{0}e^{-ru}\partial^{2}_{x,y}A(u,\tilde{X}_{u},v^{2}_{u})B_{u}\frac{\sigma_{u}}{T-u}d[W,M]_{u}\\ 
&+&\sqrt{1-\rho^{2}}\int^{T}_{0}e^{-ru}\partial^{2}_{x,y}A(u,\tilde{X}_{u},v^{2}_{u})B_{u}\frac{\sigma_{u}}{T-u}d[\tilde{W},M]_{u}\\ 
&+&\rho\int^{T}_{0}e^{-ru}\partial_{x}A(u,\tilde{X}_{u},v^{2}_{u})\sigma_{u} d[W,B]_{u}\\
&+&\sqrt{1-\rho^{2}}\int^{T}_{0}e^{-ru}\partial_{x}A(u,\tilde{X}_{u},v^{2}_{u})\sigma_{u} d[\tilde{W},B]_{u}\\ 
&+&\int^{T}_{0}e^{-ru}\partial_{y}A(u,\tilde{X}_{u},v^{2}_{u})\frac{1}{T-u}d[M,B]_{u}.
\end{eqnarray*}
After applying expectations on both sides of the equation, we end up with the statement of the theorem.
\end{proof}

% ------------------------------------------------------------------------------ end: Generic_decomposition.tex
% ------------------------------------------------------------------------------ begin: model.tex
\section{A decomposition formula for SVJ models.}\label{sec:model}

In the previous section, we have given a general decomposition formula that can be used for stochastic volatility models with continuous sample paths. In this section, we are going to extend the previous decomposition to the case of a general jump diffusion model with finite activity jumps.

The main idea, like the one used in \cite{MerinoVives17}, is to adapt the pricing process in a way to be able to apply the decomposition technique effectively. In our case, this would translate into conditioning on the finite number of jumps $n_{T}$. If we denote $J_{n}=\sum_{i=0}^{n}Y_{i}$, using the integrability of Black-Scholes function, we can obtain the following conditioning formula for European options with payoff at maturity $T : BS(T,X_{T},v_{T})$. 
\begin{eqnarray*}
V_{0}&=&e^{-rT}\mathbb{E}\left[BS(T,X_{T},v_{T})\right]\\
&=&e^{-rT}\sum^{+\infty}_{n=0} p_{n}(\lambda T) \mathbb{E}\left[BS\left(T, \tilde{X}_{T} + \sum^{n_{T}}_{i=0} Y_{i}, v_{T}\right)\Big|n_{T}=n\right]\\
&=&e^{-rT}\sum^{+\infty}_{n=0} p_{n}(\lambda T) \mathbb{E}\left[BS\left(T, \tilde{X}_{T} + J_{n}, v_{T}\right)\right]\\
&=&e^{-rT}\sum^{\infty}_{n=0} p_{n}(\lambda T) \mathbb{E}\left[\mathbb{E}_{J_{n}}\left[BS(T, \tilde{X}_{T} + J_{n}, v_{T})\right]\right]\\
&=&e^{-rT}\sum^{\infty}_{n=0} p_{n}(\lambda T) \mathbb{E}\left[G_{n}(T, \tilde{X}_{T},  v_{T})\right].
\end{eqnarray*}
where 
$$G_{n}(T, \tilde{X}_{T}, v_{T}):=\mathbb{E}_{J_{n}}\left[BS(T, \tilde{X}_{T} + J_{n}, v_{T})\right].$$

We have switched our problem from a jump diffusion model with stochastic volatility to another one with no jumps. %
Combining the generic SV decomposition formula (from Theorem \ref{Generic Deco}) and conditioning on the number of jumps we obtain a corner-stone for our approximation.
\begin{corollary}[SVJ decomposition formula]\label{BS Deco}
Let $X_t$ be a log-price process \eqref{log-model}, $G_{n}$ be the previously defined function. Then we can express the call option fair value $V_0$ using the Poisson mass function $p_n$ and a martingale process $M_t$ (defined by \eqref{def:M_t}). In particular,  
\begin{eqnarray*}
V_{0}&=&\sum^{\infty}_{n=0} p_{n}(\lambda T) G_{n}(0,\tilde{X}_{0},v_{0})\\
&+&\frac{1}{8}\sum^{\infty}_{n=0} p_{n}(\lambda T) \mathbb{E}\left[\int^{T}_{0}e^{-ru}\Gamma^{2} G_{n}(u,\tilde{X}_{u},v_{u})d[M,M]_{u}\right] \\
&+&\frac{\rho}{2}\sum^{\infty}_{n=0} p_{n}(\lambda T)\mathbb{E}\left[\int^{T}_{0}e^{-ru}\Lambda\Gamma G_{n}(u,\tilde{X}_{u},v_{u})\sigma_{u}d[W,M]_{u}\right].
\end{eqnarray*}
\end{corollary}

\begin{proof}
We apply Theorem \ref{Generic Deco} to $A(t,\tilde{X}_{t},v^{2}_{t}):=G_{n}(t,\tilde{X}_{t},v_{t})$ and $B_{t}\equiv 1$. Note that
\begin{eqnarray*}
\partial_{\sigma^{2}}BS(t,x,\sigma)=\frac{(T-t)}{2}\left(\partial^{2}_{x}-\partial_{x}\right)BS(t,x,\sigma)
\end{eqnarray*}
and
\begin{eqnarray*}
\partial^{2}_{\sigma^{2}}BS(t,x,\sigma)=\frac{(T-t)^{2}}{4}\left(\partial^{2}_{x}-\partial_{x}\right)^{2}BS(t,x,\sigma).
\end{eqnarray*}
Then, the corollary follows immediately. Note that in order to apply the It\^o formula to function $G_{n}$ we need to use a mollifier argument as it is done in \cite{MerinoVives15}.
\end{proof} 

\begin{remark}
For clarity, in the following we will refer to terms of the previous decomposition as

$$V_{0}=\sum^{\infty}_{n=0} p_{n}(\lambda T) G_{n}(0,\tilde{X}_{0},v_{0}) +\sum^{\infty}_{n=0} p_{n}(\lambda T) \left[ (I_{n}) + (II_{n}) \right].$$
\end{remark}

To compute the above expression can be cumbersome. The main idea is to find an alternative formula such that the main terms are easier to be computed while paying the price by having more terms in the formula. Fortunately, in many cases these new terms can be neglected as approximation error. The size of the error depends on the model and whether we are focusing on short or long time dynamics.\par

The following lemma is proved in Al\`os~(2012), p.~406; and will help us to derive bounds on the error terms that appear in the main result of this paper - a computationally suitable decomposition formula for generic finite activity SVJ models.
\begin{lemma}\label{lemaclau}
\label{fitage} Let $0\leq t\leq s\leq T$ and $\mathcal{G}_{t}:=\mathcal{F}_{t}\vee
\mathcal{F}_{T}^{W}.$ For every $n\geq 0,$ there exists $C=C(n) $ such that
\begin{equation*}
\left\vert \mathbb{E} \left( \left. {\Lambda^{n}\Gamma BS}\left( s,\tilde{X}_{s},v_{s}\right)
\right\vert \mathcal{G}_{t}\right) \right\vert \leq C \left( \int_{s}^{T}E_s
\left(\sigma_{\theta }^{2}\right) d\theta \right)^{-\frac{1}{2}\left(
n+1\right) }.
\end{equation*}
\end{lemma}

\begin{theorem}[Computationally suitable SVJ decomposition]
\label{Generic Approximation formula}
Let $X_t$ be a log-price process \eqref{log-model} and $G_{n}$ be the previously defined function. Then we can express the call option fair value $V_0$ using the Poisson probability mass function $p_n$ and processes $R_t, U_t$ defined by \eqref{def:R_t} and \eqref{def:U_t}, respectively. In particular,  
\begin{eqnarray*}
V_{0}&=&\sum^{\infty}_{n=0} p_{n}(\lambda T) G_{n}(0,\tilde{X}_{0},v_{0})\\
&+&\sum^{\infty}_{n=0} p_{n}(\lambda T) \Gamma^{2} G_{n}(0,\tilde{X}_{0},v_{0})R_{0} \\
&+&\sum^{\infty}_{n=0} p_{n}(\lambda T)\Lambda\Gamma G_{n}(0,\tilde{X}_{0},v_{0})U_{0}\\
&+&\sum^{\infty}_{n=0} p_{n}(\lambda T) \Omega_{n}
\end{eqnarray*}
where $\Omega_{n}$ are error terms fully derived in Appendix \ref{error_terms}. 
\end{theorem}
\begin{proof}
We use Theorem \ref{Generic Deco} iteratively for the following choices of $A(t,X_{t}, v^{2}_{t})$:
\begin{enumerate}
\item[(I):]  $$A(t,X_{t}, v^{2}_{t}):=\Gamma^{2} G_{n}(t,\tilde{X}_{t},v_{t})$$ and $$B_{t}:=R_{t}=\frac{1}{8}\mathbb{E}_{t}\left[\int^{T}_{t}d[M,M]_{u}\right].$$
\item[(II):] $$A(t,X_{t}, v^{2}_{t}):=\Lambda\Gamma G_{n}(t,\tilde{X}_{t},v_{t})$$ and $$B_{t}:=U_{t}=\frac{\rho}{2}\mathbb{E}_{t}\left[\int^{T}_{t}\sigma_{u}d[W,M]_{u}\right].$$
\end{enumerate}
and then the statement follows immediately. See also the terms in Appendix \ref{error_terms}.
\end{proof}
 As we will illustrate in the upcoming sections for Heston-type SVJ models - this formula can be efficiently evaluated, while the neglected error terms do not significantly limit a practical use of the formula. The main ingredients, to get SVJ approximate pricing formula, are expressions for $R_0, U_0$ and $G_{n}(0,\tilde{X}_{0},v_{0})$. Now we provide some insight how the latter term can be expressed under various jump-diffusion settings. 
\begin{remark}\label{rmk:G_jumps}
In particular, we have a closed formula for a log-normal jump diffusion model (e.g. \cite{Bates96} SVJ model):
$$G_{n}(0, \tilde{X}_{0},  v_{0})=BS\left(0, \tilde{X}_{0}, \sqrt{v^{2}_{0}+ n \frac{\sigma^{2}_{J}}{T}}\right)$$
where we modified the risk-free rate used in the Black-Scholes formula to
$$r^{*}=r - \lambda\left(e^{\mu_{J} + \frac{1}{2}\sigma_J^2} - 1\right) + n\frac{\mu_{J} + \frac{1}{2}\sigma_J^2}{T}.$$
A very similar formula for the Merton case is deduced by \cite{Hanson07}. More details will follow in the next sections. Under general (finite-activity) jump diffusion settings, we will need to solve 
$$\int_{\mathbb{R}} BS\left(0, \tilde{X}_{0} + y, v_{0}\right) f_{J_{n}}(y) dy$$
where $f_{J_{n}}=(f^{*n}_{Y})(y)$ is the convolution of the law of $n$ jumps.

Here we provide a list of known results for various popular models.
\begin{itemize}
\item \cite{Kou02} double exponential model:
\begin{eqnarray*}
f^{*(n)}(u)&=&e^{-\eta_{1}u}\sum^{n}_{k=1} P_{n,k}\eta^{k}_{1}\frac{1}{(k-1)!}u^{k-1}\mathbf{1}_{\left\{u\geq0\right\}}\\
&+&e^{-\eta_{2}u}\sum^{n}_{k=1} Q_{n,k}\eta^{k}_{2}\frac{1}{(k-1)!}(-u)^{k-1}\mathbf{1}_{\left\{u<0\right\}}
\end{eqnarray*}
where
\begin{eqnarray*}
P_{n,k}=\sum^{n-1}_{i=k}\binom{n-k-1}{i-k}\binom{n}{i}\binom{\eta_{1}}{\eta_{1}+\eta_{2}}^{i-k}\binom{\eta_{2}}{\eta_{1}+\eta_{2}}^{n-i} p^{i}q^{n-i}
\end{eqnarray*}
for all $1\leq k \leq n-1$, and
\begin{eqnarray*}
Q_{n,k}=\sum^{n-1}_{i=k}\binom{n-k-1}{i-k}\binom{n}{i}\binom{\eta_{1}}{\eta_{1}+\eta_{2}}^{n-i}\binom{\eta_{2}}{\eta_{1}+\eta_{2}}^{i-k} p^{n-i}q^{i}
\end{eqnarray*}
for all $1\leq k \leq n-1$. In addition, $P_{n,n}=p^{n}$ and $Q_{n,n}=q^{n}$.
\item \cite{YanHanson06} model uses log-uniform jump sizes and hence the density is of the form \citep{Killmann01}:
\begin{eqnarray*}
f^{*(n)}(u) = \begin{cases}
        \frac{\sum^{\tilde{n}(n,u)}_{i=0}(-1)^{i}\binom{n}{i}\left(u-na-i(b-a)\right)^{n-1}}{(n-1)!(b-a)^{n}} & \text{ if } na \leq u \leq nb \\
        0 & \text{otherwise.} \\
       \end{cases}
\end{eqnarray*}
where $\tilde{n}(n,u):=\left[\frac{u-na}{b-a}\right]$ is the largest integer less than $\frac{u-na}{b-a}$.\par
\end{itemize}
\end{remark}

% ------------------------------------------------------------------------------ end: model.tex
% ------------------------------------------------------------------------------ begin: heston.tex
\section{SVJ models of the Heston type}\label{sec:heston}

In this section, we apply the previous generic results to derive a pricing formula for SVJ models with the Heston variance process. The aim is not to provide pricing solution for all known/studied models, but rather to detail the derivation for a selected model and comment on possible extension to different models. I.e. we focus on models with dynamics satisfying the following stochastic differential equations 
\begin{eqnarray}\label{log-model-heston}
dX_{t}&=& \left(r - \lambda k-\frac{1}{2}\sigma^{2}_{t}\right)dt + \sigma_{t} \left(\rho dW_{t}  + \sqrt{1-\rho^{2}}d\tilde{W}_{t} \right) +dJ_{t}\\
d\sigma^{2}_{t}&=& \kappa \left(\theta - \sigma^{2}_{t}\right)dt + \nu \sqrt{\sigma^{2}_{t}}dW_{t}
\end{eqnarray}
where $\sigma_{0}$, $\kappa$, $\theta$, $\nu$ are positive constants satisfying the Feller condition $2\kappa\theta \geq \nu^{2}$. The process $\sigma_{t}^2$ represents an instantaneous variance of the price at time $t$, $\theta$ is a long run average level of the variance, $\kappa$ is a rate at which $\sigma_{t}$ reverts to $\theta$ and, last but not least, $\nu$ is a volatility of volatility parameter. We will distinguish between the two cases: 
\begin{itemize}
\item either jump amplitudes follow a Gaussian process (\cite{Bates96} model), 
\item or they are driven by other models, e.g. a log-uniform process (\cite{YanHanson06} model).
\end{itemize}

\subsection{Approximation of the SVJ models of the Heston type}

For a standard Heston model, we have the following results, see \cite{Alos15}: 
\begin{lemma}\label{remarkADV15}
Assume the standard notation from the previous sections alongside specific definitions. Define $\varphi(t):=\int_t^T e^{-\kappa(z-t)}dz.$ We have the following results:
\begin{enumerate}
\item For $s\geq t$ we have
\begin{equation*}
E_t(\sigma^2_s)=\theta+(\sigma_t^2-\theta)e^{-\kappa (s-t)}=\sigma_t^2 e^{-\kappa (s-t)}+\theta (1-e^{-\kappa (s-t)}),
\end{equation*}
so, in particular, this quantity is bounded below by $\sigma_t^2\wedge\theta$ and above by $\sigma_t^2\vee\theta$.

\item $E_t\left(\int_{t}^{T}\sigma _{s}^{2}ds\right) = \theta\left(T-t\right) + \frac{\sigma_{t}^{2}-\theta}{\kappa}\left(1-e^{-\kappa\left( T-t\right) }\right).$

\item $dM_t\ =\ \nu\sigma_t\left(\int_t^T e^{-\kappa (u-t)}du\right) dW_t= \frac{\nu}{\kappa}\,\sigma_t\left(1-e^{-\kappa(T-t)}\right)dW_t.$

\item $U_t:=\frac{\rho}{2} E_t\left(\int_{t}^{T}\sigma _{s}d\left\langle
M,W\right\rangle _{s}\right)=\frac{\rho}{2}\nu\int_{t}^{T}E_t \left(\sigma
_{s}^{2}\right) \left( \int_{s}^{T}e^{-\kappa (u-s)}du\right) ds$
\begin{eqnarray*}
\hspace{2em} =\,\frac{\rho \nu}{2\kappa^{2}}\left\{\theta \kappa \left(
T-t\right) -2\theta +\sigma _{t}^{2}+e^{-\kappa \left( T-t\right) }\left(
2\theta -\sigma _{t}^{2}\right) -\kappa \left( T-t\right) e^{-\kappa \left(
T-t\right) }\left( \sigma _{t}^{2}-\theta \right) \right\}.
\end{eqnarray*}

\item $R_t:=\frac{1}{8}E_t\left(\int_{t}^{T}d\left\langle M,M\right\rangle_s\right) \
=\ \frac{1}{8}\nu ^{2}\int_{t}^{T}E_t\left(\sigma _s^2\right)
\left( \int_s^T e^{-\kappa (u-s)}du\right) ^{2}ds$
\begin{eqnarray*}
&=&\frac{\nu ^{2}}{8\kappa ^{2}}\left\{ \theta \left( T-t\right)
+\frac{\left(\sigma _{t}^{2}-\theta \right)}{\kappa}\left(1-e^{-\kappa \left(T-t\right)}\right)\right. \\
\\
&&\ -\frac{2\theta }{\kappa }\left( 1-e^{-\kappa \left( T-t\right) }\right)
-2\left( \sigma _{t}^{2}-\theta \right) \left( T-t\right) e^{-\kappa \left(T-t\right) } \\
\\
&&\ +\left. \frac{\theta }{2\kappa }\left( 1-e^{-2\kappa \left( T-t\right)}\right)
+\frac{\left( \sigma _{t}^{2}-\theta \right) }{\kappa }
\left(e^{-\kappa(T-t)}-e^{-2\kappa \left( T-t\right) }\right) \right\}.
\end{eqnarray*}

\item $dU_t = \frac{\rho \nu^2}{2}\left(\int_t^T e^{-\kappa(z-t)}\varphi(z)dz\right)\sigma_t dW_t-\frac{\rho\nu}{2}\varphi(t)\sigma^2_t dt$,

\item $dR_t = \frac{\nu^3}{8}\left(\int_t^T e^{-\kappa(z-t)}\varphi(z)^2dz\right) \sigma_t dW_t-\frac{\nu^2}{8}\varphi(t)^2 \sigma_t^2 dt$.
\end{enumerate}
\end{lemma}

Furthermore, the following lemma is proved in \cite{Alos15}.

\begin{lemma}\label{tercerlema} Let all the objects be well defined as above, then for a standard Heston model we have that\\
\begin{tabular}{cl}
(i)  &   $\int_s^T E_s(\sigma_u^2) du\  \geq\  \frac{\theta\kappa}{2}\left(\int_s^T e^{-\kappa (u-s)} du\right)^2$, \\ \\
(ii) &  $\int_{s}^{T}E_{s}\left( \sigma _u^2\right) du\  \geq\  \sigma _{s}^{2}\left(\int_{s}^{T}e^{-\kappa (u-s)}du\right)$.
\end{tabular}
\end{lemma}

\begin{remark}
We can utilize these equalities to get analogue results for Theorem \ref{Generic Approximation formula}. The $\Omega_{n}$ terms can be founded in Appendix \ref{error_terms_HestonSVJ}.%A.2.
\end{remark}

Now we have all the tools needed to introduce the main practical result - pricing formula 
\begin{corollary}[Heston-type SVJ pricing formula]
\label{SVJ models of the Heston type Approximation} 
Let $G_{n}(0,\tilde{X}_{0},v_{0})$ takes the expression as in Remark \ref{rmk:G_jumps} for a particular jump-type setting, let 
\begin{eqnarray*}
R_0&=&\frac{\nu ^{2}}{8\kappa ^{2}}\left\{ \theta T
+\frac{\left(\sigma _{0}^{2}-\theta \right)}{\kappa}\left(1-e^{-\kappa T}\right)-\frac{2\theta }{\kappa }\left( 1-e^{-\kappa T }\right)
-2\left( \sigma _{0}^{2}-\theta \right)  T e^{-\kappa T } \right. \\
\\
&&\ +\left. \frac{\theta }{2\kappa }\left( 1-e^{-2\kappa T}\right)
+\frac{\left( \sigma _{0}^{2}-\theta \right) }{\kappa }
\left(e^{-\kappa T}-e^{-2\kappa T}\right) \right\}
\end{eqnarray*}
and let 
\begin{eqnarray*}
U_0 =\,\frac{\rho \nu}{2\kappa^{2}}\left\{\theta \kappa
T -2\theta +\sigma _{0}^{2}+e^{-\kappa T }\left(
2\theta -\sigma _{0}^{2}\right) -\kappa T e^{-\kappa
T }\left( \sigma _{0}^{2}-\theta \right) \right\}.
\end{eqnarray*}
 Then the European option fair value is expressed as

\begin{eqnarray*}
V_{0}&=&\sum^{\infty}_{n=0} p_{n}(\lambda T) G_{n}(0,\tilde{X}_{0},v_{0})\\
&+&\sum^{\infty}_{n=0} p_{n}(\lambda T) \Gamma^{2} G_{n}(0,\tilde{X}_{0},v_{0})R_{0} \\
&+&\sum^{\infty}_{n=0} p_{n}(\lambda T)\Lambda\Gamma G_{n}(0,\tilde{X}_{0},v_{0})U_{0}\\
&+&\sum^{\infty}_{n=0} p_{n}(\lambda T)\Omega_{n}
\end{eqnarray*}
where $\Omega_{n}$ are error terms detailed in Appendix \ref{error_terms_HestonSVJ} . The upper bound for any $\Omega_{n}$ is given by
\begin{eqnarray*}
\Omega_{n} \leq  \nu^{2}(\left|\rho\right| + \nu)^{2} \left( \frac{1}{r}\wedge (T-t)\right) \Pi(\kappa,\theta)
\end{eqnarray*}
where $\Pi(\kappa,\theta)$ is a positive function. Therefore, the total error 
\begin{eqnarray*}
\Omega=\sum^{\infty}_{n=0} p_{n}(\lambda T)\Omega_{n}
\end{eqnarray*}
is bounded by the same constant.
\end{corollary}

\begin{proof}
We plug-in the Heston volatility model dynamics into Theorem \ref{Generic Approximation formula}. Using the integrability of the Black-Scholes function, Fubini Theorem and the fact that the upper bound of Lemma \ref{lemaclau} does not depend on the log spot price, the upper bound can be used for every $G_{n}$ function. Using Lemma \ref{remarkADV15} and Lemma \ref{tercerlema} we prove the corollary. The whole proof is in Appendix \ref{UB_HestonSVJ}.%A.3.
\end{proof}

\begin{remark}[Approximate fractional SVJ model]
For the model introduced by \cite{PospisilSobotka16amf} one can derive a very similar decomposition as in Corollary \ref{SVJ models of the Heston type Approximation}. In fact, only the terms $R_0$ and $U_0$ have to be changed while the other terms remain the same.
\end{remark}

\subsection{Numerical analysis of the SVJ models of the Heston type}\label{Sec:NumHeston}
In this section, we compare the newly obtained approximation formula for option prices under \cite{Bates96} model (i.e. log-normal jump sizes alongside Heston model's instantaneous variance) with the market standard approach for pricing European options under SVJ models - the Fourier-transform based pricing formula. The comparison is performed with two important aspects in mind:
\begin{itemize}
\item practical precision of the pricing formula when neglecting the total error term $\Omega$,
\item efficiency of the formula expressed in terms of the computational time needed for particular pricing tasks.
\end{itemize}

In particular, we utilize a semi-closed form solution with one numerical integration as a reference price \citep{BaustianMrazekPospisilSobotka17asmb} alongside a classical solution derived by \cite{Bates96}\footnote{With a slight modification mentioned in \cite{Gatheral06} to not suffer the "Heston trap" issues.}. The numerical integration errors according to \cite{BaustianMrazekPospisilSobotka17asmb} should be typically well beyond $10^{-10}$, hence we can take the numerically computed prices as the reference prices for the comparison.

Due to the theoretical properties of the total error term $\Omega$, we illustrate the approximation quality for several values of $\rho$ and $\nu$ while keeping other parameters fixed\footnote{The considered model and market parameters take the following values: $S_0=100$; $r=0.001$; $\tau=0.3$; $v_0=0.25$; $\kappa=1.5$; $\theta=0.2$; $\lambda = 0.05$; $\mu_J = -0.05$; $\sigma_J = 0.5$.}.      
\begin{figure}
\includegraphics[width=0.95\textwidth]{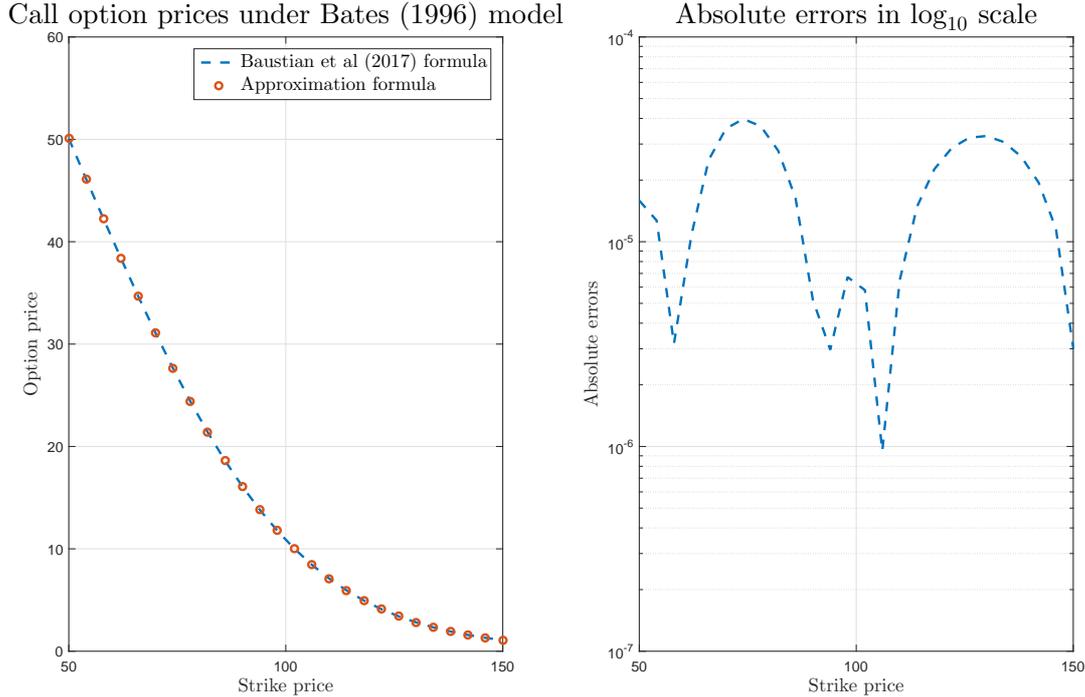}
\caption{Approximation and reference prices for $\rho=-0.2$, $\nu= 5\% $ and $\tau = 0.3$.}
\label{Fig:lowNuRho}
\end{figure}

In Figure \ref{Fig:lowNuRho}, we inspect a mode of low volatility of the spot variance $\nu$ and low absolute value of the instantaneous correlation  $\rho$ between the two Brownian motions. The errors for an option price smile that corresponds to $\tau = 0.3$ are within $10^{-4}-10^{-6}$ range, while slightly better absolute errors were obtained at-the-money. Increasing either the absolute value of $\rho$ or volatility $\nu$ should, in theory, worsen the computed error measures. However, if only one of the values is increased we are still able to keep the errors below $10^{-3}$  in most of the cases, see Figure \ref{Fig:lowNuhighRho}.
\begin{figure}
\includegraphics[width=0.95\textwidth]{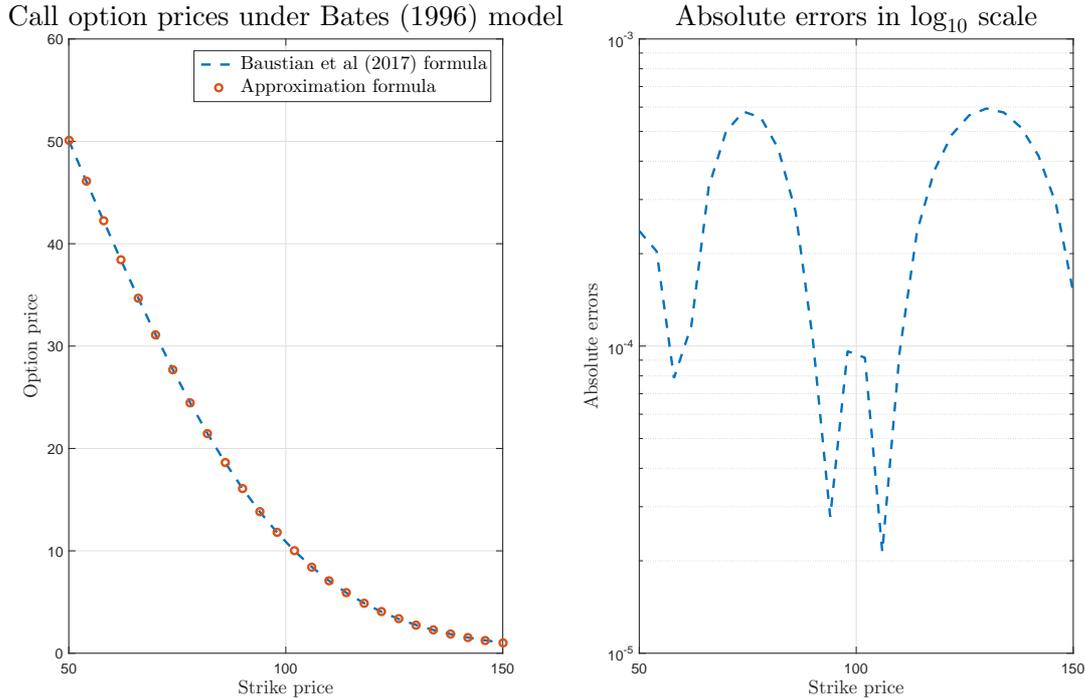}
\caption{Approximation and reference prices for $\rho=-0.8$, $\nu= 5\% $ and $\tau = 0.3$.}
\label{Fig:lowNuhighRho}
\end{figure}

Last but not least, we illustrate the approximation quality for parameters that are not well suited for the approximation. This is done by setting $\nu = 50\%$, correlation $\rho=-0.8$ and a smile with respect to $\tau = 3$. The obtained errors are depicted by Figure \ref{Fig:highNURho}. Despite the values of parameters, the shape of the option price curve remains fairly similar to the one obtained by a more precise semi-closed formula.

Main advantage of the proposed pricing approximation lies in its computational efficiency -- which might be advantageous for many tasks in quantitative finance that need fast evaluation of derivative prices. To inspect the time consumption we set up three pricing tasks. We use a batch of $100$ call options with different strikes and times to maturities that involves all types of options\footnote{It includes OTM, ATM, ITM options with short-, mid- and long-term times to maturities}. In the first task, we evaluate prices for the batch with respect to $100$ (uniformly) randomly sampled parameter sets. This should encompass a similar number of price evaluations as a market calibration task with a very good initial guess. Further on, we repeat the same trials only for $1000$ and $10000$ parameter sets, to mimic the number of evaluations for a typical local-search calibration and a global-search calibration respectively, for more information about calibration tasks see e.g. \cite{MikhailovNogel03} and \cite{MrazekPospisilSobotka16ejor}.  

\begin{figure}
\includegraphics[width=0.95\textwidth]{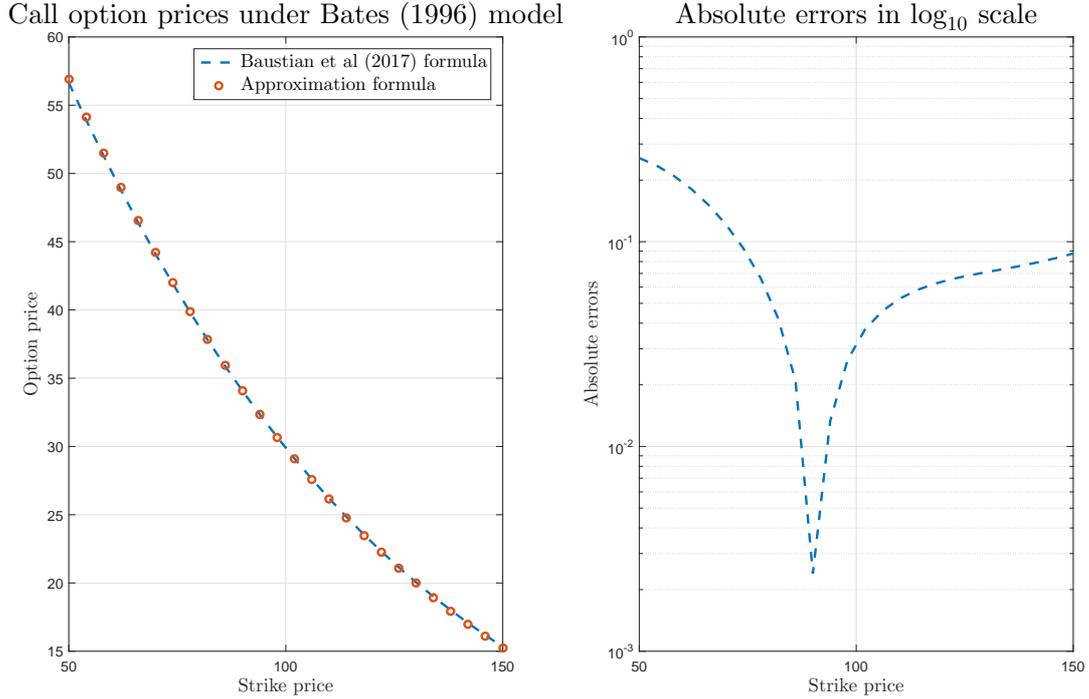}
\caption{Approximation and reference prices for $\rho=-0.8$, $\nu= 50\% $ and $\tau = 3$.}
\label{Fig:highNURho}
\end{figure}

The obtained computational times are listed in Table \ref{Tab:Res1}. Unlike the formulas with numerical integration, the proposed approximation has almost linear dependency of computational time on the number of evaluated prices. Also the results vary based on the randomly generated parameter values for numerical schemes much more than for the approximation -- this is caused by adaptivity of numerical quadratures that were used\footnote{For both \cite{BaustianMrazekPospisilSobotka17asmb} and \cite{Gatheral06} formulas we use an adaptive  Gauss-Kronrod(7,15) quadrature.}. The newly proposed approximation is typically $3\times$ faster compared to the classical two integral pricing formula and the computational time consumption does not depend on the model- nor on market-parameters.

\begin{table}[ht]
\caption{Efficiency of the Bates SVJ pricing formulas}\label{Tab:Res1}
\centering
\begin{tabular}{lcrc}
Pricing approach    & Task & Time$^\dagger$ [sec]  & Speed-up factor\\
\toprule
\multirow{3}{*}{Approximation formula}     &  $\#1$    & $0.97$ &   $3.23\times$  \\
          & $\# 2$         & $10.03$  &   $2.94\times$ \\
          & $\# 3$     & $99.67$  &   $2.83\times$\\
\midrule
\multirow{3}{*}{\cite{BaustianMrazekPospisilSobotka17asmb}}     &  $\#1$    & $2.09$ &   $1.52\times$  \\
          & $\# 2$         & $17.28$  &   $1.71\times$  \\
          & $\# 3$     & $135.95$  &   $2.01\times$ \\
\midrule
\multirow{3}{*}{\cite{Gatheral06}}     & $\# 1$    & $3.18$  & -  \\
          & $\# 2$         & $29.48$ &     - \\
          & $\# 3$     & $281.72$  &   - \\
\bottomrule
\end{tabular}\\
$^\dagger$ {\footnotesize The results were obtained on a PC with Intel Core i7-6500U CPU and 8 GB RAM.}

\end{table}

\par

% ------------------------------------------------------------------------------ end: heston.tex
% ------------------------------------------------------------------------------ begin: implied.tex
\section{The approximated implied volatility surface for SVJ models of the Heston type}\label{sec:implied}

In the above section, we have computed a bound for the error between the exact price and the approximated pricing formula for the SVJ models of the Heston type. Now, we are going to derive an approximation of the implied volatility surface alongside the corresponding ATM implied volatility profiles. These approximations can help us to understand the volatility dynamics of studied models in a better way.
 
\subsection{Deriving an approximated implied volatility surface for SVJ models of the Heston type}

The price of an European call option with strike $K$ and maturity $T$ is an observable quantity which will be referred to
as $P_0^{obs}=P^{obs}(K,T)$. Recall that the \emph{implied volatility} is defined as the value $I(T,K)$ that satisfies
\begin{equation*}
BS(0,S_0,I(T,K))= P_{0}^{obs}.
\end{equation*}

Using the results from the previous section, we are going to derive an approximation to the implied volatility as in \cite{Fouque03a}. We use the idea to expand the implied volatility function $I(T,K)$ with respect to two scales. For illustration of the idea, we recall that according to asymptotic sequences $\{\delta^k\}_{k=0}^\infty$, $\{\epsilon^k\}_{k=0}^\infty$ converging to $0$ are considered. Thus, we can write
\[
f=f_{0,0}+\delta f_{1,0} + \epsilon f_{0,1} + O((\delta + \epsilon)^{2}),
\]
for a particular function $f$. Let $\epsilon=\rho\nu$ and $\delta=\nu^{2}$, then we expand $I(T,K)$ with respect to these two scales as
\begin{eqnarray*}
I(T,K)= v_{0} + \rho\nu I_{1}(T,K) + \nu^{2}I_{2}(T,K) + O((\rho \nu + \nu^{2})).
\end{eqnarray*}
We will denote by $\hat{I}(T,K)= v_{0} + \rho\nu I_{1}(T,K) + \nu^{2}I_{2}(T,K)$ the approximation to the implied volatility and by $\hat{V}(0,x,v_{0})$ the approximation to the option price which was obtained in Corollary  \ref{SVJ models of the Heston type Approximation}. We know that according to Corollary  \ref{SVJ models of the Heston type Approximation}: 
\begin{eqnarray*}
\hat{V}(0,x,v_{0})&=&\sum^{\infty}_{n=0} p_{n}(\lambda T) BS(0,x+J_{n},v_{0})\\
&+&\sum^{\infty}_{n=0} p_{n}(\lambda T) \Gamma^{2} BS(0,x+J_{n},v_{0})R_{0} \\
&+&\sum^{\infty}_{n=0} p_{n}(\lambda T)\Lambda\Gamma BS(0,x+J_{n},v_{0})U_{0}.
\end{eqnarray*}
To simplify the notation, we define
\begin{eqnarray*}
\gamma_{n}:=\frac{d^{2}_{+}(x,r,\sigma)-d^{2}_{+}(x+J_{n},r,\sigma)}{2}
\end{eqnarray*}
and
\begin{eqnarray*}
D_{1}(x, J_{n}, \sigma, T)&:=&\mathbb{E}_{J_{n}}\left[\frac{e^{J_{n}+\gamma_{n}}}{\sigma T}\left(1-\frac{d_+(x+J_{n}, r,\sigma)}{\sigma \sqrt{T}}\right)\right],\\
D_{2}(x, J_{n}, \sigma, T)&:=&\mathbb{E}_{J_{n}}\left[\frac{e^{J_{n}+\gamma_{n}}}{\sigma^{3} T^{2}}\left(d^{2}_{+}(x+J_{n},r, \sigma)-\sigma d_+(x+J_{n},r,\sigma) \sqrt{T}-1\right)\right].
\end{eqnarray*}
Using the fact that
\begin{eqnarray*}
\partial_{\sigma}BS(t,x,\sigma)=\frac{e^{x}e^{-d^{2}_{+}(\sigma)/2}\sqrt{T-t}}{\sqrt{2\pi}},
\end{eqnarray*}
we can re-write the approximated price as
\begin{eqnarray*}
\hat{V}(0,x,v_{0})&=&\sum^{\infty}_{n=0} p_{n}(\lambda T) BS(0,x+J_{n},v_{0})\\
&+& \partial_{\sigma}BS(v_{0}) \sum^{\infty}_{n=0} p_{n}(\lambda T)D_{1}(x, J_{n}, \sigma, T)U_0\\
&+& \partial_{\sigma}BS(v_{0}) \sum^{\infty}_{n=0} p_{n}(\lambda T)  D_{2}(x, J_{n}, v_{0}, T)R_0.
\end{eqnarray*}
where we write $BS(v_0)$ as a shorthand for $BS(0,x,v_0)$. Consider now the Taylor expansion of $BS(0,x,I(T,K))$ around $v_0$:
\begin{eqnarray}
BS(0,x,I(T,K))& = & BS(v_0) + \partial_{\sigma} B(v_0)(\rho\nu I_1(T,K) + \nu^2 I_2(T,K) + \cdots) \nonumber \\
& & +\frac{1}{2}\partial^{2}_{\sigma} BS(v_0)(\rho\nu I_1(T,K) + \nu^2 I_2(T,K) + \cdots)^2
+ \ \cdots  \nonumber \\ \nonumber
& = & BS(v_0) + \rho\nu\partial_{\sigma} BS(v_0) I_1(T,K) + \nu^2\partial_{\sigma} BS(v_0) I_2(T,K) \
+ \ \cdots. \label{expansionBS}
\end{eqnarray}

Noticing that
$$BS(v_0)=\sum^{\infty}_{n=0} p_{n}(\lambda T) BS(0,x+J_{n},v_{0}) $$
and equating
$$\hat{V}(0,x,v_{0})=BS(0,x,I(T,K)),$$
we obtain
\begin{eqnarray}
\hat{I}_{1}(T,K)&:=&\rho\nu I_1(T,K) \text{ } = \text{ } U_{0} \sum^{\infty}_{n=0} p_{n}(\lambda T)D_{1}(x, J_{n}, v_{0}, T),\\
\hat{I}_{2}(T,K)&:=&\nu^2 I_2(T,K)\text{ } = \text{ } R_{0} \sum^{\infty}_{n=0} p_{n}(\lambda T)  D_{2}(x, J_{n}, v_{0}, T).
\end{eqnarray}

Hence, we have the following approximation of implied volatility
\begin{eqnarray*}
\hat{I}(T,K)&=&v_{0} +U_{0}\sum^{\infty}_{n=0} p_{n}(\lambda T)D_{1}(x, J_{n}, v_{0}, T)\\
&+& R_{0}\sum^{\infty}_{n=0} p_{n}(\lambda T)  D_{2}(x, J_{n}, v_{0}, T).
\end{eqnarray*}

In particular, when we look at the ATM curve, we have that
\begin{eqnarray*}
\hat{I}^{ATM}(T)&=&v_{0} +  U_{0} \sum^{\infty}_{n=0} p_{n}(\lambda T)\mathbb{E}_{J_{n}}\left[\frac{e^{J_{n}+\gamma_{n}}}{v_{0} T}\left(\frac{1}{2} - \frac{J_{n}}{T v^{2}_{0}}\right)\right]\\
&-&R_{0}\sum^{\infty}_{n=0} p_{n}(\lambda T)  \mathbb{E}_{J_{n}}\left[\frac{e^{J_{n}+\gamma_{n}}}{v_{0} T}\left(\frac{1}{4}+\frac{1}{v^{2}T}-\frac{J^{2}_{n}}{v^{4}_{0}T^{2}}\right)\right].
\end{eqnarray*}

\begin{remark}
When $T$ converges to $0$, the dynamics of the model is the same as in the Heston model. This is due to the behavior of the Poisson process when $T\downarrow 0$.
\end{remark}
  
\subsection{Deriving an approximated implied volatility surface for Bates model}
The Bates model is a particular example of SVJ model of the Heston type. The fact that jumps are also log-normal makes the model more tractable. In this section, we will adapt the generic formulas to this particular case. In this model, 
after each jump, the drift- and volatility-like parameters will change. We define 
$$\tilde{v}^{(n)}_{0}=\sqrt{v^{2}_{0} + n\frac{\sigma^{2}_{J}}{T}}$$
as the new volatility and
$$\tilde{r}_{n}=r- \lambda\left(e^{\mu_{J} + \frac{1}{2}\sigma_J^2} - 1\right) + n\frac{\mu_{J} + \frac{1}{2}\sigma_J^2}{T}$$
as the new drift. The parameter $n$ is the number of realized jumps, $\mu_{J}$ and $\sigma_{J}$ are the jump-size parameters and $\lambda$ is the jump intensity. For simplicity, we denote:
$$c_{n}:=- \lambda\left(e^{\mu_{J} + \frac{1}{2}\sigma_J^2} - 1\right) + n\frac{\mu_{J} + \frac{1}{2}\sigma_J^2}{T}.$$

As a consequence, we have that
\begin{eqnarray*}
d_{\pm}\left(x,\tilde{r}_{n},\tilde{v}^{(n)}_{0}\right)= \frac{x-\ln K+ \tilde{r}_{n}T}{\tilde{v}^{(n)}_{0} \sqrt{T}} \pm \frac{\tilde{v}^{(n)}_{0} \sqrt{T}}{2}.
\end{eqnarray*}
Following the steps done in the generic formula, we can define the variables
\begin{eqnarray*}
D_{B,1}\left(x,\tilde{r}_{n},\tilde{v}^{(n)}_{0}, T\right)&=&\frac{e^{\gamma_{n}}}{\tilde{v}^{(n)}_0 T}\left(1-\frac{d_+\left(x,\tilde{r}_{n},\tilde{v}^{(n)}_0\right)}{\tilde{v}^{(n)}_0 \sqrt{T}}\right),\\
D_{B,2}\left(x,\tilde{r}_{n}, \tilde{v}^{(n)}_{0}, T\right)&=&\frac{e^{\gamma_{n}}}{\tilde{v}^{(n)}_0 T}\left(\frac{d_+^2\left(x,\tilde{r}_{n},\tilde{v}^{(n)}_0\right)-\tilde{v}^{(n)}_0 d_+\left(x,\tilde{r}_{n},\tilde{v}^{(n)}_0\right) \sqrt{T}-1}{\left(\tilde{v}^{(n)}_0\right)^2 T} \right).
\end{eqnarray*}
It follows that
\begin{eqnarray}
\hat{I}_{B,1}(T,K)&=&\rho\nu I_{B,1}(T,K) \text{ } = \text{ } U_0\sum^{\infty}_{n=0} p_{n}(\lambda T)D_{B,1}\left(x,\tilde{r}_{n},\tilde{v}^{(n)}_{0}, T\right),\\
\hat{I}_{B,2}(T,K)&=&\nu^2 I_{B,2}(T,K) \text{ }  = \text{ } R_0\sum^{\infty}_{n=0} p_{n}(\lambda T) D_{B,2}\left(x,\tilde{r}_{n}, \tilde{v}^{(n)}_{0}, T\right).
\end{eqnarray}
The approximation of the implied volatility surface has the following shape
\begin{eqnarray*}
\hat{I}_{B}(T,K)&=&v_{0}+U_0 \sum^{\infty}_{n=0}p_{n}(\lambda T)\frac{e^{\gamma_{n}}}{\tilde{v}^{(n)}_{0} T}\left(1-\frac{d_+\left(x,\tilde{r}_{n},\tilde{v}^{(n)}_{0}\right)}{\tilde{v}^{(n)}_{0} \sqrt{T}}\right)\\
&+&R_0 \sum^{\infty}_{n=0}p_{n}(\lambda T)\frac{e^{\gamma_{n}}}{\tilde{v}^{(n)}_{0} T}\left(\frac{d_+^2\left(x,\tilde{r}_{n},\tilde{v}^{(n)}_{0}\right)-\tilde{v}^{(n)}_{0} d_+\left(x,\tilde{r}_{n},\tilde{v}^{(n)}_{0}\right) \sqrt{T}-1}{\left(\tilde{v}^{(n)}_{0}\right)^2 T}\right).
\end{eqnarray*}
In particular, the ATM implied volatility curve under the studied model takes the form:
\begin{eqnarray*}
\hat{I}^{ATM}_{B}(T)&=&v_{0}+U_0 \sum^{\infty}_{n=0}p_{n}(\lambda T)\frac{e^{\gamma^{ATMBates}_{n}}}{\tilde{v}^{(n)}_{0} T}\left(\frac{1}{2}- \frac{c_{n}}{\left(\tilde{v}^{(n)}_{0}\right)^{2}}\right)\\
&-&R_0 \sum^{\infty}_{n=0}p_{n}(\lambda T)\frac{e^{\gamma^{ATMBates}_{n}}}{\tilde{v}^{(n)}_{0} T}\left(\frac{1}{4}+\frac{1}{\left(\tilde{v}^{(n)}_{0}\right)^{2}T}-\frac{c^{2}_{n}}{\left(\tilde{v}^{(n)}_{0}\right)^{4}}\right)
\end{eqnarray*}
where 
\begin{eqnarray*}
\gamma^{ATMBates}_{n}= -\frac{1}{2}\left(c_{n}T +\frac{c^{2}_{n}T}{\left(\tilde{v}^{(n)}_{0}\right)^{2}}\right).
\end{eqnarray*}

\subsection{Numerical analysis of the approximation of the implied volatility for the Bates case}
\begin{figure}
\includegraphics[width=0.95\textwidth]{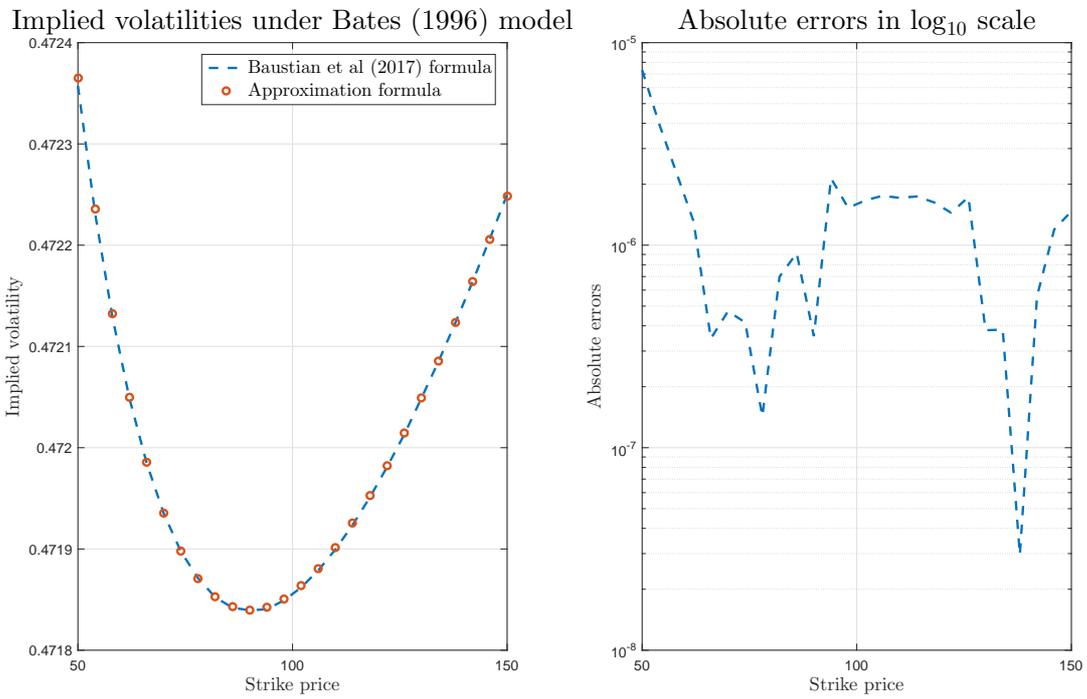}
\caption{Approximation and reference implied volatilities for $\rho=-0.2$, $\nu= 5\% $ and $\tau = 0.3$.}
\label{Fig:lowNuRho_I}
\end{figure}

In the previous section we have compared the approximation and semi-closed form formulas for option prices under \cite{Bates96} model. For this model, we also illustrate the approximation quality in terms of implied volatilities.

Because there is no exact closed formula for implied volatilities under the studied model, we take as a reference price the one obtained by means of the complex Fourier transform \citep{BaustianMrazekPospisilSobotka17asmb}. Once we have computed the prices 
we use a numerical inversion to obtain the desired implied volatilities.

As previously, we start by comparing implied volatilities for well-suited parameter sets. The illustration in Figure \ref{Fig:lowNuRho_I} is obtained by setting $\rho = -0.1$, $\nu = 5\%$ and other parameters as in Section \ref{Sec:NumHeston}. Typically, for a well-suited parameter set, the absolute approximation errors stay within the range $10^{-5}-10^{-7}$.

\begin{figure}
\includegraphics[width=0.95\textwidth]{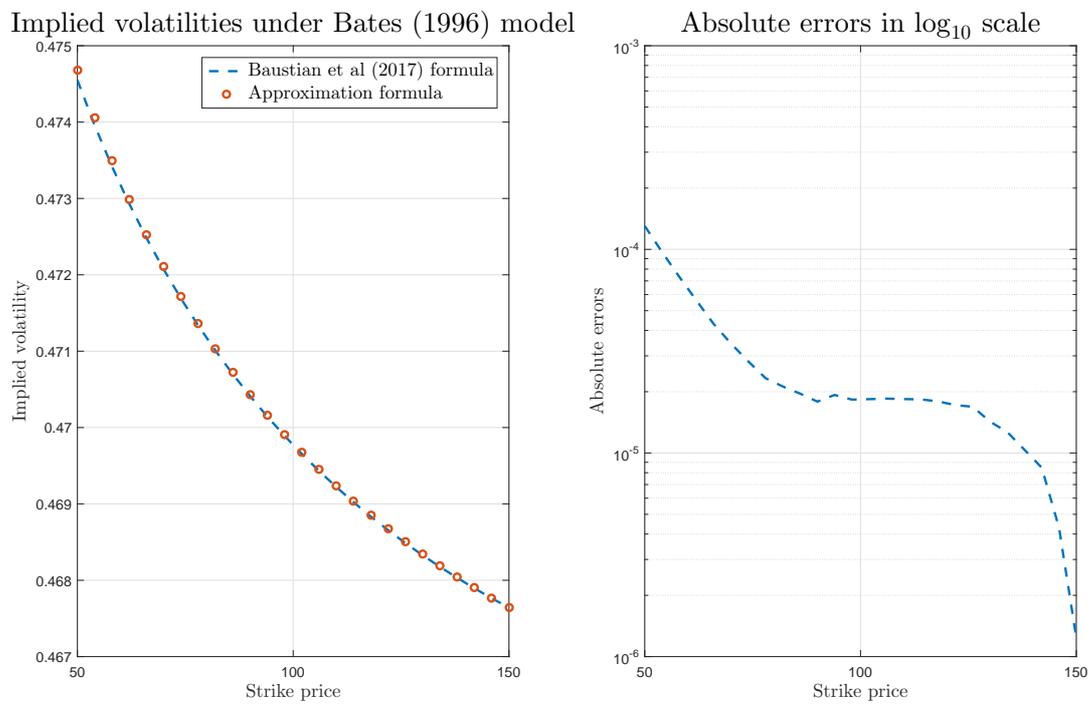}
\caption{Approximation and reference implied volatilities for $\rho=-0.8$, $\nu= 5\% $ and $\tau = 0.3$.}
\label{Fig:lowNuhighRho_I}
\end{figure}

Even for not entirely well-suited parameters we are able to obtain reasonable errors especially for ATM options, see Figures \ref{Fig:lowNuhighRho_I} and \ref{Fig:highNURho_I}. In the mode of high volatility $\nu$ of the variance process and high absolute value of the instantaneous correlation $\rho$, the curvature of the smile is not fully captured. However, the errors are typically well below $10^{-2}$ even in this adverse setting.

\begin{figure}
\includegraphics[width=0.95\textwidth]{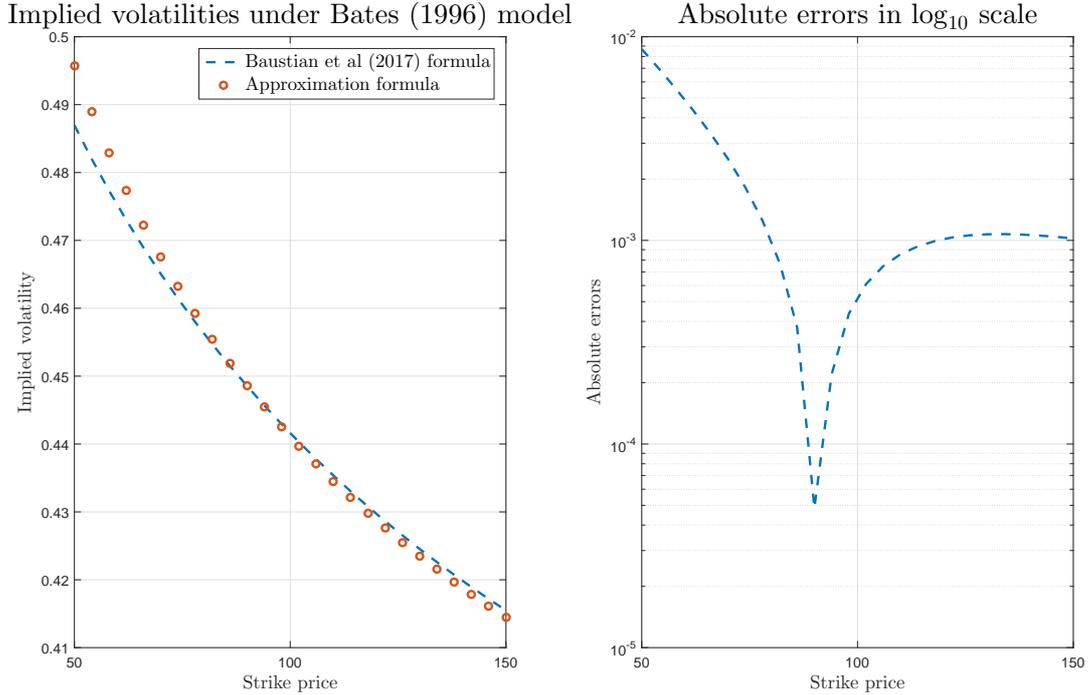}
\caption{Approximation and reference implied volatilities for $\rho=-0.8$, $\nu= 50\% $ and $\tau = 3$.}
\label{Fig:highNURho_I}
\end{figure}

% ------------------------------------------------------------------------------ end: implied.tex
% ------------------------------------------------------------------------------ begin: conclusion.tex
\section{Conclusion}\label{sec:conclusion}

The aim of the paper was to derive a generic decomposition formula for SVJ option pricing models with finite activity jumps. In Section \ref{sec:model} we derived this decomposition by extending the results obtained by \cite{Alos12} for \cite{Heston93} SV model. Newly obtained decomposition is rather versatile since it does not need to specify the underlying volatility process and only common integrability and specific sample path properties are required. 

Particular approximation formulas for several SVJ models were presented in Section \ref{sec:heston} together with the numerical comparison for the \cite{Bates96} model for which we showed that the newly proposed approximation is typically three times faster compared to the classical two integral semi-closed pricing formula. Moreover, its computational time does not depend on the model parameters nor on market data. The biggest advantage of the proposed pricing approximation therefore lies in its computational efficiency, which is advantageous for many tasks in quantitative finance such as calibration to real market data that can lead to an extensive number of formula evaluations for SVJ models. On the other hand, general decomposition formula allowed us to understand the key terms contributing to the option fair value under specific models and hence this theoretical result has also its practical impact. 

In Section \ref{sec:implied}, we have obtained an approximated volatility surface under SVJ models and we provided a boundary case simplification for ATM options. In particular, we have studied the approximation in the \cite{Bates96} model case. A numerical comparison of this approximation is also presented. 

Although the generic approach covers various interesting SVJ models, there are other models that do not fit into the general structure described in Section \ref{sec:preliminaries}. For these models, such as \cite{BarndorffNielsen01} model or infinite activity jumps models, we still might be able to derive a similar decomposition, that was beyond the scope of the present paper. Newly obtained results therefore give suggestions on how to derive approximation formulas for other models.

% ------------------------------------------------------------------------------ end: conclusion.tex
% ------------------------------------------------------------------------------ begin: acknowledgement.tex
\section*{Acknowledgements}

This work was partially supported by the GACR Grant GA18-16680S Rough models of fractional stochastic volatility.
Computational resources were provided by the CESNET LM2015042 and the CERIT Scientific Cloud LM2015085, provided under the programme "Projects of Large Research, Development, and Innovations Infrastructures".

The work of Josep Vives is partially supported by Spanish grant MEC MTM 2016-76420-P. 

% ------------------------------------------------------------------------------ end: acknowledgement.tex
\appendix%\normalsize
% ------------------------------------------------------------------------------ begin: appendix.tex
\section{Appendices}\label{sec:appendix}

In the following appendices we obtain the error terms of the decomposition in Theorem \ref{Generic Approximation formula} (Appendix \ref{error_terms}), the same formulas for the SVJ model of the Heston type (Appendix \ref{error_terms_HestonSVJ}) and upper bounds for those terms using Corollary \ref{SVJ models of the Heston type Approximation} (Appendix \ref{UB_HestonSVJ}). 
\subsection{ Decomposition formulas in the general model}\label{error_terms}
In this section, we obtain the error terms for a general model.
\subsubsection{Decomposition of the term ($I_{n}$)}
The term I can be decomposed by
\begin{eqnarray*}
&&\frac{1}{8}\mathbb{E}\left[\int^{T}_{0}e^{-ru}\Gamma^{2} G_{n}(u,\tilde{X}_{u},v_{u})d[M,M]_{u}\right]-\Gamma^{2} G_{n}(0,\tilde{X}_{0},v_{0})R_{0}\\ 
&=&\frac{1}{8}\mathbb{E}\left[\int^{T}_{0}e^{-ru}\Gamma^{4} G_{n}(u,\tilde{X}_{u},v_{u})R_{u}d[M,M]_{u}\right]\\ 
&+&\frac{\rho}{2}\mathbb{E}\left[\int^{T}_{0}e^{-ru}\Lambda \Gamma^{3} G_{n}(u,\tilde{X}_{u},v_{u})R_{u}\sigma_{u}d[W,M]_{u}\right]\\ 
&+&\rho\mathbb{E}\left[\int^{T}_{0}e^{-ru}\Lambda\Gamma^{2} G_{n}(u,\tilde{X}_{u},v_{u})\sigma_{u} d[W,R]_{u}\right]\\ 
&+&\frac{1}{2}\mathbb{E}\left[\int^{T}_{0}e^{-ru}\Gamma^{3} G_{n}(u,\tilde{X}_{u},v_{u})d[M,R]_{u}\right].
\end{eqnarray*}

\subsubsection{Decomposition of the term ($II_{n}$)}
The term II can be decomposed by
\begin{eqnarray*}
&&\frac{\rho}{2}\mathbb{E}\left[\int^{T}_{0}e^{-ru}\Lambda \Gamma G_{n}(u,\tilde{X}_{u},v_{u})\sigma_{u}d[W,M]_{u}\right]-\Lambda \Gamma G_{n}(0,\tilde{X}_{0},v_{0})U_{0}\\
&=&\frac{1}{8}\mathbb{E}\left[\int^{T}_{0}e^{-ru}\Lambda \Gamma^{3} G_{n}(u,\tilde{X}_{u},v_{u})U_{u}d[M,M]_{u}\right]\\ 
&+&\frac{\rho}{2}\mathbb{E}\left[\int^{T}_{0}e^{-ru}\Lambda^{2} \Gamma^{2} G_{n}(u,\tilde{X}_{u},v_{u})U_{u}\sigma_{u}d[W,M]_{u}\right]\\ 
&+&\rho\mathbb{E}\left[\int^{T}_{0}e^{-ru}\Lambda^{2} \Gamma G_{n}(u,\tilde{X}_{u},v_{u})\sigma_{u} d[W,U]_{u}\right]\\ 
&+&\frac{1}{2}\mathbb{E}\left[\int^{T}_{0}e^{-ru}\Lambda \Gamma^{2} G_{n}(u,\tilde{X}_{u},v_{u})d[M,U]_{u}\right].
\end{eqnarray*}

\subsection{Decomposition formulas in the general model for the SVJ models of the Heston type}\label{error_terms_HestonSVJ}

In this section, we obtain the error terms for the SVJ models of the Heston type.

\subsubsection{Decomposition of the term ($I_{n}$) in the SVJ models of the Heston type}
The term I can be decomposed by
\begin{eqnarray*}
&&\frac{1}{8}\mathbb{E}\left[\int^{T}_{0}e^{-ru}\Gamma^{2} G_{n}(u,\tilde{X}_{u},v_{u})d[M,M]_{u}\right]-
\frac{\nu^{2}}{8}\Gamma^{2} G_{n}(0,\tilde{X}_{0},v_{0})\left(\int^{T}_{0}\mathbb{E}\left(\sigma^2_s\right) \varphi(s)^{2} ds\right)\\ 
&=&\frac{\nu^{4}}{64}\mathbb{E}\left[\int^{T}_{0}e^{-ru}\Gamma^{4} G_{n}(u,\tilde{X}_{u},v_{u})\left(\int^{T}_{u}\mathbb{E}_{u}\left(\sigma^2_s\right) \varphi(s)^{2} ds\right)\sigma^{2}_{u}\varphi^{2}(u) du\right]\\ 
&+&\frac{\rho\nu^{3}}{16}\mathbb{E}\left[\int^{T}_{0}e^{-ru}\Lambda \Gamma^{3} G_{n}(u,\tilde{X}_{u},v_{u})\left(\int^{T}_{u}\mathbb{E}_{u}\left(\sigma^2_s\right) \varphi(s)^{2} ds\right)\sigma^{2}_{u}\varphi(u)du\right]\\ 
&+&\frac{\rho\nu^3}{8}\mathbb{E}\left[\int^{T}_{0}e^{-ru}\Lambda\Gamma^{2} G_{n}(u,\tilde{X}_{u},v_{u})\left(\int^{T}_{u} e^{-\kappa(z-u)}\varphi(z)^2dz\right) \sigma^{2}_{u} du\right]\\ 
&+&\frac{\nu^4}{16}\mathbb{E}\left[\int^{T}_{0}e^{-ru}\Gamma^{3} G_{n}(u,\tilde{X}_{u},v_{u})\left(\int^{T}_{u} e^{-\kappa(z-u)}\varphi(z)^2dz\right) \varphi(u)\sigma^{2}_{u}du\right].
\end{eqnarray*}

\subsubsection{Decomposition of the term ($II_{n}$) in the SVJ models of the Heston type}
The term II can be decomposed by
\begin{eqnarray*}
&&\frac{\rho}{2}\mathbb{E}\left[\int^{T}_{0}e^{-ru}\Lambda \Gamma G_{n}(u,\tilde{X}_{u},v_{u})\sigma_{u}d[W,M]_{u}\right]-\frac{\rho\nu}{2}\Lambda \Gamma G_{n}(0,\tilde{X}_{0},v_{0})\left(\int^{T}_{0}\mathbb{E} \left(\sigma_{s}^{2}\right)  \varphi(s) ds\right)\\
&=&\frac{\rho\nu^{3}}{16}\mathbb{E}\left[\int^{T}_{0}e^{-ru}\Lambda \Gamma^{3} G_{n}(u,\tilde{X}_{u},v_{u})\left(\int^{T}_{u}\mathbb{E}_{u} \left(\sigma_{s}^{2}\right)  \varphi(s) ds\right)\sigma^{2}_{u}\varphi(u)^{2} du\right]\\ 
&+&\frac{\rho^{2}\nu^{2}}{4}\mathbb{E}\left[\int^{T}_{0}e^{-ru}\Lambda^{2} \Gamma^{2} G_{n}(u,\tilde{X}_{u},v_{u})\left(\int^{T}_{u}\mathbb{E}_{u} \left(\sigma_{s}^{2}\right)  \varphi(s) ds\right)\sigma^{2}_{u}\varphi(u) du\right]\\ 
&+&\frac{\rho^{2} \nu^2}{2}\mathbb{E}\left[\int^{T}_{0}e^{-ru}\Lambda^{2} \Gamma G_{n}(u,\tilde{X}_{u},v_{u})\left(\int^{T}_{u} e^{-\kappa(z-u)}\varphi(z)dz\right)\sigma^{2}_{u} du\right]\\ 
&+&\frac{\rho \nu^3}{4}\mathbb{E}\left[\int^{T}_{0}e^{-ru}\Lambda \Gamma^{2} G_{n}(u,\tilde{X}_{u},v_{u})\left(\int^{T}_{u} e^{-\kappa(z-u)}\varphi(z)dz\right)\sigma^{2}_u\varphi(u) du\right].
\end{eqnarray*}

\subsection{Upper-Bound of decomposition formulas in the SVJ models of the Heston type}\label{UB_HestonSVJ}
In this section, we obtain the upper-bounds for the SVJ models of the Heston type.
\subsubsection{Upper-Bound of the term ($I_{n}$) in the SVJ models of the Heston type}
We can re-write the decomposition formula as
\begin{eqnarray*}
&&\frac{1}{8}\mathbb{E}\left[\int^{T}_{0}e^{-r(u-t)}\Gamma^{2} G_{n}(u,\tilde{X}_{u},v_{u})d[M,M]_{u}\right]-
\frac{\nu^{2}}{8}\Gamma^{2} G_{n}(0,\tilde{X}_{0},v_{0})\left(\int^{T}_{0}E\left(\sigma^2_s\right) \varphi(s)^{2} ds\right)\\ 
&=&\frac{\nu^{4}}{64}\mathbb{E}\left[\int^{T}_{0}e^{-ru}\left(\partial^{6}_{x}-3\partial^{5}_{x}+3\partial^{4}_{x}- \partial^{3}_{x}\right)\Gamma G_{n}(u,\tilde{X}_{u},v_{u})\left(\int^{T}_{u}\mathbb{E}_{u}\left(\sigma^2_s\right) \varphi(s)^{2} ds\right)\sigma^{2}_{u}\varphi^{2}(u) du\right]\\ 
&+&\frac{\rho\nu^{3}}{16}\mathbb{E}\left[\int^{T}_{0}e^{-ru} \left(\partial^{5}_{x}-2\partial^{4}_{x}+\partial^{3}_{x}\right) \Gamma G_{n}(u,\tilde{X}_{u},v_{u})\left(\int^{T}_{u}\mathbb{E}_{u}\left(\sigma^2_s\right) \varphi^{2}(s) ds\right)\sigma^{2}_{u}\varphi(u)du\right]\\ 
&+&\frac{\rho\nu^3}{8}\mathbb{E}\left[\int^{T}_{0}e^{-ru}\left(\partial^{3}_{x}-\partial^{2}_{x}\right) \Gamma G_{n}(u,\tilde{X}_{u},v_{u})\left(\int^{T}_{u} e^{-\kappa(z-u)}\varphi(z)^2dz\right) \sigma^{2}_{u} du\right]\\ 
&+&\frac{\nu^4}{16}\mathbb{E}\left[\int^{T}_{0}e^{-ru}\left(\partial^{4}_{x}-2\partial^{3}_{x}+\partial^{2}_{x}\right)\Gamma G_{n}(u,\tilde{X}_{u},v_{u})\left(\int^{T}_{u} e^{-\kappa(z-u)}\varphi(z)^2dz\right) \varphi(u)\sigma^{2}_{u}du\right].
\end{eqnarray*}
Applying Lemma \ref{lemaclau} and defining $a_{u}:=v_{u}\sqrt{T-u}$, we obtain
\begin{eqnarray*}
&&\left|\frac{1}{8}\mathbb{E}\left[\int^{T}_{0}e^{-ru}\Gamma^{2} G_{n}(u,\tilde{X}_{u},v_{u})d[M,M]_{u}\right]-
\frac{\nu^{2}}{8}\Gamma^{2} G_{n}(0,\tilde{X}_{0},v_{0})\left(\int^{T}_{0}\mathbb{E}\left(\sigma^2_s\right) \varphi(s)^{2} ds\right)\right|\\ 
&\leq&C\frac{\nu^{4}}{64}\mathbb{E}\left[\int^{T}_{0}e^{-ru}\left(\frac{1}{a^7_{u}}+\frac{3}{a^6_{u}}+\frac{3}{a^5_{u}}+\frac{1}{a^4_{u}}\right)v^{2}_{u} (T-u) \varphi(u)^{4}\sigma^{2}_{u} du\right]\\ 
&+&C\frac{\left|\rho\right|\nu^{3}}{16}\mathbb{E}\left[\int^{T}_{0}e^{-ru} \left(\frac{1}{a^6_{u}}+\frac{2}{a^5_{u}}+\frac{1}{a^4_{u}}\right)v^{2}_{u} (T-u) \varphi(u)^{3}\sigma^{2}_{u}du\right]\\ 
&+&C\frac{\left|\rho\right|\nu^3}{8}\mathbb{E}\left[\int^{T}_{0}e^{-ru}\left(\frac{1}{a^4_{u}}+\frac{1}{a^3_{u}}\right)\sigma^{2}_{u}\varphi(u)^3  du\right]\\ 
&+&C\frac{\nu^4}{16}\mathbb{E}\left[\int^{T}_{0}e^{-ru}\left(\frac{1}{a^5_{u}}+\frac{2}{a^4_{u}}+\frac{1}{a^3_{u}}\right) \varphi(u)^{4}\sigma^{2}_{u}du\right].
\end{eqnarray*}
Now, using Lemma \ref{tercerlema} (ii), we have  
\begin{eqnarray*}
&&\left|\frac{1}{8}\mathbb{E}\left[\int^{T}_{0}e^{-ru}\Gamma^{2} G_{n}(u,\tilde{X}_{u},v_{u})d[M,M]_{u}\right]-
\frac{\nu^{2}}{8}\Gamma^{2} G_{n}(0,\tilde{X}_{0},v_{0})\left(\int^{T}_{0}\mathbb{E}\left(\sigma^2_s\right) \varphi^{2}(s) ds\right)\right|\\ 
&\leq&C\frac{\nu^{4}}{64}\mathbb{E}\left[\int^{T}_{0}e^{-ru}\left(\frac{1}{a^7_{u}}+\frac{3}{a^6_{u}}+\frac{3}{a^5_{u}}+\frac{1}{a^4_{u}}\right)v^{4}_{u} (T-u)^{2} \varphi(u)^{3} du\right]\\ 
&+&C\frac{\left|\rho\right|\nu^{3}}{16}\mathbb{E}\left[\int^{T}_{0}e^{-ru} \left(\frac{1}{a^6_{u}}+\frac{2}{a^5_{u}}+\frac{1}{a^4_{u}}\right)v^{4}_{u} (T-u)^{2} \varphi(u)^{2}du\right]\\ 
&+&C\frac{\left|\rho\right|\nu^3}{8}\mathbb{E}\left[\int^{T}_{0}e^{-ru}\left(\frac{1}{a^4_{u}}+\frac{1}{a^{3}_{u}}\right)v^{2}_{u}(T-u)\varphi(u)^2 du\right]\\ 
&+&C\frac{\nu^4}{16}\mathbb{E}\left[\int^{T}_{0}e^{-ru}\left(\frac{1}{a^5_{u}}+\frac{2}{a^4_{u}}+\frac{1}{a^{3}_{u}}\right) \varphi(u)^{3}v^{2}_{u}(T-u)du\right].
\end{eqnarray*}
Finally, applying Lemma \ref{tercerlema} (i), we find that 
\begin{eqnarray*}
&&\left|\frac{1}{8}\mathbb{E}\left[\int^{T}_{0}e^{-ru}\Gamma^{2} G_{n}(u,\tilde{X}_{u},v_{u})d[M,M]_{u}\right]-
\frac{\nu^{2}}{8}\Gamma^{2} G_{n}(0,\tilde{X}_{0},v_{0})\left(\int^{T}_{0}\mathbb{E}\left(\sigma^2_s\right) \varphi(s)^{2} ds\right)\right|\\ 
&\leq&C\frac{\nu^{4}}{64}\mathbb{E}\left[\int^{T}_{0}e^{-ru}\left(\frac{2\sqrt{2}}{\theta \kappa\sqrt{\theta \kappa}} + \frac{6}{\theta \kappa^{2}} + \frac{3\sqrt{2}}{\kappa^{2}\sqrt{\theta \kappa}} + \frac{1}{\kappa^{3}}\right)   du\right]\\ 
&+&C\frac{\left|\rho\right|\nu^{3}}{16}\mathbb{E}\left[\int^{T}_{0}e^{-ru} \left(\frac{2}{\theta \kappa} + \frac{2\sqrt{2}}{\kappa\sqrt{\theta \kappa}} + \frac{1}{\kappa^{2}}\right) du\right]\\ 
&+&C\frac{\left|\rho\right|\nu^3}{8}\mathbb{E}\left[\int^{T}_{0}e^{-ru}\left(\frac{2}{\theta \kappa} + \frac{\sqrt{2}}{\kappa\sqrt{\theta \kappa}}\right) du\right]\\ 
&+&C\frac{\nu^4}{16}\mathbb{E}\left[\int^{T}_{0}e^{-ru}\left(\frac{2\sqrt{2}}{\theta \kappa\sqrt{\theta \kappa}} + \frac{4}{\theta \kappa^{2}} + \frac{\sqrt{2}}{\kappa^{2}\sqrt{\theta \kappa}}\right) du\right].
\end{eqnarray*}
Then we have that
\begin{eqnarray*}
&&\left|\frac{1}{8}\mathbb{E}\left[\int^{T}_{0}e^{-ru}\Gamma^{2} G_{n}(u,\tilde{X}_{u},v_{u})d[M,M]_{u}\right]-
\frac{\nu^{2}}{8}\Gamma^{2} G_{n}(0,\tilde{X}_{0},v_{0})\left(\int^{T}_{0}\mathbb{E}\left(\sigma^2_s\right) \varphi(s)^{2} ds\right)\right|\\ 
&\leq&C\frac{\nu^{4}}{64}\left(\frac{2\sqrt{2}}{\theta \kappa\sqrt{\theta \kappa}} + \frac{6}{\theta \kappa^{2}} + \frac{3\sqrt{2}}{\kappa^{2}\sqrt{\theta \kappa}} + \frac{1}{\kappa^{3}}\right)\left(\int^{T}_{0}e^{-ru}   du\right)\\ 
&+&C\frac{\left|\rho\right|\nu^{3}}{16}\left(\frac{2}{\theta \kappa} + \frac{2\sqrt{2}}{\kappa\sqrt{\theta \kappa}} + \frac{1}{\kappa^{2}}\right)\left(\int^{T}_{0}e^{-ru}  du\right)\\ 
&+&C\frac{\left|\rho\right|\nu^3}{8}\left(\frac{2}{\theta \kappa} + \frac{\sqrt{2}}{\kappa\sqrt{\theta \kappa}}\right)\left(\int^{T}_{0}e^{-ru} du\right)\\ 
&+&C\frac{\nu^4}{16}\left(\frac{2\sqrt{2}}{\theta \kappa\sqrt{\theta \kappa}} + \frac{4}{\theta \kappa^{2}} + \frac{\sqrt{2}}{\kappa^{2}\sqrt{\theta \kappa}}\right)\left(\int^{T}_{0}e^{-ru} du\right).
\end{eqnarray*}
Using the fact that $\int_{t}^{T}e^{-ru}ds\leq \frac{1}{r}\wedge T$, we conclude that
\begin{eqnarray*}
&&\left|\frac{1}{8}\mathbb{E}\left[\int^{T}_{0}e^{-ru}\Gamma^{2} G_{n}(u,\tilde{X}_{u},v_{u})d[M,M]_{u}\right]-
\frac{\nu^{2}}{8}\Gamma^{2} G_{n}(0,\tilde{X}_{0},v_{0})\left(\int^{T}_{0}\mathbb{E}\left(\sigma^2_s\right) \varphi(s)^{2} ds\right)\right|\\ 
&\leq&  \nu^3\left(\left|\rho\right| + \nu\right)\left( \frac{1}{r}\wedge T\right) \Pi_{1}(\kappa,\theta)
\end{eqnarray*}
where $\Pi_{1}$ is a positive function.

\subsubsection{Upper-Bound of the term ($II_{n}$) in the SVJ models of the Heston type}
We can re-write the decomposition formula as
\begin{eqnarray*}
&&\frac{\rho}{2}\mathbb{E}\left[\int^{T}_{0}e^{-ru}\Lambda \Gamma G_{n}(u,\tilde{X}_{u},v_{u})\sigma_{u}d[W,M]_{u}\right]-\frac{\rho\nu}{2}\Lambda \Gamma G_{n}(0,\tilde{X}_{0},v_{0})\left(\int^{T}_{0}\mathbb{E} \left(\sigma_{s}^{2}\right)  \varphi(s) ds\right)\\
&=&\frac{\rho\nu^{3}}{16}\mathbb{E}\left[\int^{T}_{0}e^{-ru}\left(\partial^{5}_{x}-2\partial^{4}_{x}+\partial^{3}_{x}\right) \Gamma G_{n}(u,\tilde{X}_{u},v_{u})\left(\int^{T}_{u}\mathbb{E}_{u} \left(\sigma_{s}^{2}\right)  \varphi(s) ds\right)\sigma^{2}_{u}\varphi^{2}(u) du\right]\\ 
&+&\frac{\rho^{2}\nu^{2}}{4}\mathbb{E}\left[\int^{T}_{0}e^{-ru}\left(\partial^{4}_{x}-\partial^{3}_{x}\right) \Gamma G_{n}(u,\tilde{X}_{u},v_{u})\left(\int^{T}_{u}\mathbb{E}_{u} \left(\sigma_{s}^{2}\right)  \varphi(s) ds\right)\sigma^{2}_{u}\varphi(u) du\right]\\ 
&+&\frac{\rho^{2} \nu^2}{2}\mathbb{E}\left[\int^{T}_{0}e^{-ru}\partial^{2}_{x}\Gamma G_{n}(u,\tilde{X}_{u},v_{u})\left(\int^{T}_{u} e^{-\kappa(z-u)}\varphi(z)dz\right)\sigma^{2}_{u} du\right]\\ 
&+&\frac{\rho \nu^3}{4}\mathbb{E}\left[\int^{T}_{0}e^{-ru}\left(\partial^{3}_{x}-\partial^{2}_{x}\right)\Gamma G_{n}(u,\tilde{X}_{u},v_{u})\left(\int^{T}_{u} e^{-\kappa(z-u)}\varphi(z)dz\right)\sigma^{2}_u\varphi(u) du\right]
\end{eqnarray*}
Applying Lemma \ref{lemaclau} and defining $a_{u}:=v_{u}\sqrt{T-u}$, we obtain
\begin{eqnarray*}
&&\left|\frac{\rho}{2}\mathbb{E}\left[\int^{T}_{0}e^{-ru}\Lambda \Gamma G_{n}(u,\tilde{X}_{u},v_{u})\sigma_{u}d[W,M]_{u}\right]-\frac{\rho\nu}{2}\Lambda \Gamma G_{n}(0,\tilde{X}_{0},v_{0})\left(\int^{T}_{0}\mathbb{E} \left(\sigma_{s}^{2}\right)  \varphi(s) ds\right)\right|\\
&\leq&C\frac{\left|\rho\right|\nu^{3}}{16}\mathbb{E}\left[\int^{T}_{0}e^{-ru}\left(\frac{1}{a^6_{u}}+\frac{2}{a^5_{u}}+\frac{1}{a^4_{u}}\right)\left(\int^{T}_{u}\mathbb{E}_{u} \left(\sigma_{s}^{2}\right)  \varphi(s) ds\right)\sigma^{2}_{u}\varphi(u)^{2} du\right]\\ 
&+&C\frac{\rho^{2}\nu^{2}}{4}\mathbb{E}\left[\int^{T}_{0}e^{-ru}\left(\frac{1}{a^5_{u}}+\frac{1}{a^4_{u}}\right) \left(\int^{T}_{u}\mathbb{E}_{u} \left(\sigma_{s}^{2}\right)  \varphi(s) ds\right)\sigma^{2}_{u}\varphi(u) du\right]\\ 
&+&C\frac{\rho^{2} \nu^2}{2}\mathbb{E}\left[\int^{T}_{0}e^{-ru}\frac{1}{a^3_{u}} \left(\int^{T}_{u} e^{-\kappa(z-u)}\varphi(z)dz\right)\sigma^{2}_{u} du\right]\\ 
&+&C\frac{\left|\rho \right|\nu^3}{4}\mathbb{E}\left[\int^{T}_{0}e^{-ru}\left(\frac{1}{a^4_{u}}+\frac{1}{a^3_{u}}\right) \left(\int^{T}_{u} e^{-\kappa(z-u)}\varphi(z)dz\right)\sigma^{2}_u\varphi(u) du\right].
\end{eqnarray*}
Using Lemma \ref{tercerlema} (ii), then
\begin{eqnarray*}
&&\left|\frac{\rho}{2}\mathbb{E}\left[\int^{T}_{0}e^{-ru}\Lambda \Gamma G_{n}(u,\tilde{X}_{u},v_{u})\sigma_{u}d[W,M]_{u}\right]-\frac{\rho\nu}{2}\Lambda \Gamma G_{n}(0,\tilde{X}_{0},v_{0})\left(\int^{T}_{0}\mathbb{E} \left(\sigma_{s}^{2}\right)  \varphi(s) ds\right)\right|\\
&\leq&C\frac{\left|\rho\right|\nu^{3}}{16}\mathbb{E}\left[\int^{T}_{0}e^{-ru}\left(\frac{1}{a^6_{u}}+\frac{2}{a^5_{u}}+\frac{1}{a^4_{u}}\right)v^{4}_{u}(T-u)^{2}\varphi(u)^{2} du\right]\\ 
&+&C\frac{\rho^{2}\nu^{2}}{4}\mathbb{E}\left[\int^{T}_{0}e^{-ru}\left(\frac{1}{a^5_{u}}+\frac{1}{a^4_{u}}\right) v^{4}_{u}(T-u)^{2}\varphi(u) du\right]\\ 
&+&C\frac{\rho^{2} \nu^2}{2}\mathbb{E}\left[\int^{T}_{0}e^{-ru}\frac{1}{a^3_{u}} \varphi(u)v^{2}_{u}(T-u) du\right]\\ 
&+&C\frac{\left|\rho\right| \nu^3}{4}\mathbb{E}\left[\int^{T}_{0}e^{-ru}\left(\frac{1}{a^4_{u}}+\frac{1}{a^3_{u}}\right) \varphi(u)^{2}v^{2}_{u}(T-u) du\right].
\end{eqnarray*}
Finally, applying Lemma \ref{tercerlema} (i), we find that 
\begin{eqnarray*}
&&\left|\frac{\rho}{2}\mathbb{E}\left[\int^{T}_{0}e^{-ru}\Lambda \Gamma G_{n}(u,\tilde{X}_{u},v_{u})\sigma_{u}d[W,M]^{c}_{u}\right]-\frac{\rho\nu}{2}\Lambda \Gamma G_{n}(0,\tilde{X}_{0},v_{0})\left(\int^{T}_{0}\mathbb{E} \left(\sigma_{s}^{2}\right)  \varphi(s) ds\right)\right|\\
&\leq&C\frac{\left|\rho\right|\nu^{3}}{16}\mathbb{E}\left[\int^{T}_{0}e^{-ru}\left(\frac{2}{\theta \kappa} + \frac{2\sqrt{2}}{\kappa\sqrt{\theta \kappa}} +  \frac{1}{\kappa^{2}}\right) du\right]\\ 
&+&C\frac{\rho^{2}\nu^{2}}{4}\mathbb{E}\left[\int^{T}_{0}e^{-ru}\left(\frac{\sqrt{2}}{\sqrt{\theta \kappa}} + \frac{1}{\kappa}\right) du\right]\\ 
&+&C\frac{\rho^{2} \nu^2}{2}\mathbb{E}\left[\int^{T}_{0}e^{-ru}\frac{\sqrt{2}}{\sqrt{\theta \kappa}}  du\right]\\ 
&+&C\frac{\left|\rho\right| \nu^3}{4}\mathbb{E}\left[\int^{T}_{0}e^{-ru}\left(\frac{2}{\theta \kappa} + \frac{\sqrt{2}}{\kappa\sqrt{\theta \kappa}}\right) du\right].
\end{eqnarray*}
Then we have that
\begin{eqnarray*}
&&\left|\frac{\rho}{2}\mathbb{E}\left[\int^{T}_{0}e^{-ru}\Lambda \Gamma G_{n}(u,\tilde{X}_{u},v_{u})\sigma_{u}d[W,M]^{c}_{u}\right]-\frac{\rho\nu}{2}\Lambda \Gamma G_{n}(0,\tilde{X}_{0},v_{0})\left(\int^{T}_{0}\mathbb{E} \left(\sigma_{s}^{2}\right)  \varphi(s) ds\right)\right|\\
&\leq&C\frac{\left|\rho\right|\nu^{3}}{16}\left(\frac{2}{\theta \kappa} + \frac{2\sqrt{2}}{\kappa\sqrt{\theta \kappa}} +  \frac{1}{\kappa^{2}}\right)\left(\int^{T}_{0}e^{-ru} du\right)\\ 
&+&C\frac{\rho^{2}\nu^{2}}{4}\left(\frac{\sqrt{2}}{\sqrt{\theta \kappa}} + \frac{1}{\kappa}\right)\left(\int^{T}_{0}e^{-ru} du\right)\\ 
&+&C\frac{\rho^{2} \nu^2}{2}\frac{\sqrt{2}}{\sqrt{\theta \kappa}}\left(\int^{T}_{0}e^{-ru}  du\right)\\ 
&+&C\frac{\left|\rho\right| \nu^3}{4}\left(\frac{2}{\theta \kappa} + \frac{\sqrt{2}}{\kappa\sqrt{\theta \kappa}}\right)\left(\int^{T}_{0}e^{-ru} du\right).
\end{eqnarray*}
Using the fact that $\int_{t}^{T}e^{-ru}ds\leq \frac{1}{r}\wedge T$, we conclude that
\begin{eqnarray*}
&&\left|\frac{\rho}{2}\mathbb{E}\left[\int^{T}_{0}e^{-ru}\Lambda \Gamma G_{n}(u,\tilde{X}_{u},v_{u})\sigma_{u}d[W,M]^{c}_{u}\right]-\frac{\rho\nu}{2}\Lambda \Gamma G_{n}(0,\tilde{X}_{0},v_{0})\left(\int^{T}_{0}\mathbb{E} \left(\sigma_{s}^{2}\right)  \varphi(s) ds\right)\right|\\
&\leq&\left|\rho\right|\nu^2 \left(\left|\rho\right| + \nu\right)\left( \frac{1}{r}\wedge T\right) \Pi_{2}(\kappa,\theta)
\end{eqnarray*}
where $\Pi_{2}$ is a positive function.

\subsubsection{Upper-Bound for the terms ($I_{n}$) and ($II_{n}$) in the SVJ models of the Heston type}
We have that 
\begin{eqnarray*}
&&\left|\frac{1}{8}\mathbb{E}\left[\int^{T}_{0}e^{-ru}\Gamma^{2} G_{n}(u,\tilde{X}_{u},v_{u})d[M,M]_{u}\right]-
\frac{\nu^{2}}{8}\Gamma^{2} G_{n}(0,\tilde{X}_{0},v_{0})\left(\int^{T}_{0}\mathbb{E}\left(\sigma^2_s\right) \varphi(s)^{2} ds\right)\right|\\ 
&+&\left|\frac{\rho}{2}\mathbb{E}\left[\int^{T}_{0}e^{-ru}\Lambda \Gamma G_{n}(u,\tilde{X}_{u},v_{u})\sigma_{u}d[W,M]^{c}_{u}\right]-\frac{\rho\nu}{2}\Lambda \Gamma G_{n}(0,\tilde{X}_{0},v_{0})\left(\int^{T}_{0}\mathbb{E} \left(\sigma_{s}^{2}\right)  \varphi(s) ds\right)\right|\\
&\leq& \nu^3\left(\left|\rho\right| + \nu\right)\left( \frac{1}{r}\wedge T\right) \Pi_{1}(\kappa,\theta) +  \left|\rho\right|\nu^2 \left(\left|\rho\right| + \nu\right)\left( \frac{1}{r}\wedge T\right) \Pi_{2}(\kappa,\theta)\\
&\leq& \nu^{2}\left(\left|\rho\right| + \nu\right)^{2}\left( \frac{1}{r}\wedge T\right) \Pi(\kappa,\theta).
\end{eqnarray*}
where function $\Pi$ is the maximum of functions $\Pi_{1}$ and $\Pi_{2}$.
% ------------------------------------------------------------------------------ end: appendix.tex
\clearpage

% ------------------------------------------------------------------------------ begin: references-preprint.tex

% ------------------------------------------------------------------------------ end: references-preprint.tex


\begin{thebibliography}{37}
\providecommand{\natexlab}[1]{#1}
\providecommand{\url}[1]{\texttt{#1}}
\providecommand{\urlprefix}{URL }
\expandafter\ifx\csname urlstyle\endcsname\relax
  \providecommand{\doi}[1]{doi:\discretionary{}{}{}#1}\else
  \providecommand{\doi}{doi:\discretionary{}{}{}\begingroup
  \urlstyle{rm}\Url}\fi
\providecommand{\eprint}[2][]{\url{#2}}

\bibitem[{Albrecher, Mayer, Schoutens, and Tistaert(2007)}]{Albrecher07}
\textsc{Albrecher, H., Mayer, P., Schoutens, W., and Tistaert, J.} (2007).
\newblock \emph{The little {H}eston trap}.
\newblock Wilmott Magazine 2007(January/February), 83--92.

\bibitem[{Al{\`o}s(2006)}]{Alos06}
\textsc{Al{\`o}s, E.} (2006).
\newblock \emph{A generalization of the {H}ull and {W}hite formula with
  applications to option pricing approximation}.
\newblock Finance Stoch. 10(3), 353--365.
\newblock ISSN 0949-2984.
\newblock \doi{10.1007/s00780-006-0013-5}.

\bibitem[{Al{\`o}s(2012)}]{Alos12}
\textsc{Al{\`o}s, E.} (2012).
\newblock \emph{A decomposition formula for option prices in the {H}eston model
  and applications to option pricing approximation}.
\newblock Finance Stoch. 16(3), 403--422.
\newblock ISSN 0949-2984.
\newblock \doi{10.1007/s00780-012-0177-0}.

\bibitem[{Al{\`o}s, de~Santiago, and Vives(2015)}]{Alos15}
\textsc{Al{\`o}s, E., de~Santiago, R., and Vives, J.} (2015).
\newblock \emph{Calibration of stochastic volatility models via second-order
  approximation: The {H}eston case}.
\newblock Int. J. Theor. Appl. Finance 18(6), 1--31.
\newblock ISSN 0219-0249.
\newblock \doi{10.1142/S0219024915500363}.

\bibitem[{Al\`os, Le\'on, Pontier, and Vives(2008)}]{Alos08}
\textsc{Al\`os, E., Le\'on, J.~A., Pontier, M., and Vives, J.} (2008).
\newblock \emph{A {H}ull and {W}hite formula for a general stochastic
  volatility jump-diffusion model with applications to the study of the
  short-time behavior of the implied volatility}.
\newblock J. Appl. Math. Stoch. Anal. pp. Art. ID 359142, 17.
\newblock ISSN 1048-9533.
\newblock \doi{10.1155/2008/359142}.

\bibitem[{Al{\`{o}}s, Le{\'{o}}n, and Vives(2007)}]{Alos07}
\textsc{Al{\`{o}}s, E., Le{\'{o}}n, J.~A., and Vives, J.} (2007).
\newblock \emph{On the short-time behavior of the implied volatility for
  jump-diffusion models with stochastic volatility}.
\newblock Finance Stoch. 11(4), 571--589.
\newblock ISSN 0949-2984.
\newblock \doi{10.1007/s00780-007-0049-1}.

\bibitem[{Ball and Roma(1994)}]{BallRoma94}
\textsc{Ball, C.~A. and Roma, A.} (1994).
\newblock \emph{Stochastic volatility option pricing}.
\newblock J. Finan. Quant. Anal. 29(4), 589–607.
\newblock ISSN 0022-1090.
\newblock \doi{10.2307/2331111}.

\bibitem[{Barndorff-Nielsen and Shephard(2001)}]{BarndorffNielsen01}
\textsc{Barndorff-Nielsen, O.~E. and Shephard, N.} (2001).
\newblock \emph{Non-{G}aussian {O}rnstein--{U}hlenbeck-based models and some of
  their uses in financial economics}.
\newblock J. R. Stat. Soc. Ser. B. Stat. Methodol. 63(2), 167--241.
\newblock ISSN 1467-9868.
\newblock \doi{10.1111/1467-9868.00282}.

\bibitem[{Bates(1996)}]{Bates96}
\textsc{Bates, D.~S.} (1996).
\newblock \emph{Jumps and stochastic volatility: Exchange rate processes
  implicit in {D}eutsche mark options}.
\newblock Rev. Financ. Stud. 9(1), 69--107.
\newblock \doi{10.1093/rfs/9.1.69}.

\bibitem[{Baustian, Mr\'{a}zek, Posp\'{\i}\v{s}il, and
  Sobotka(2017)}]{BaustianMrazekPospisilSobotka17asmb}
\textsc{Baustian, F., Mr\'{a}zek, M., Posp\'{\i}\v{s}il, J., and Sobotka, T.}
  (2017).
\newblock \emph{Unifying pricing formula for several stochastic volatility
  models with jumps}.
\newblock Appl. Stoch. Models Bus. Ind. 33(4), 422--442.
\newblock ISSN 1524-1904.
\newblock \doi{10.1002/asmb.2248}.

\bibitem[{Bayer, Friz, and Gatheral(2016)}]{Bayer15}
\textsc{Bayer, C., Friz, P., and Gatheral, J.} (2016).
\newblock \emph{Pricing under rough volatility}.
\newblock Quant. Finance 16(6), 887--904.
\newblock ISSN 1469-7688.
\newblock \doi{10.1080/14697688.2015.1099717}.

\bibitem[{Benhamou, Gobet, and Miri(2010)}]{Benhamou10}
\textsc{Benhamou, E., Gobet, E., and Miri, M.} (2010).
\newblock \emph{Time dependent {H}eston model}.
\newblock SIAM J. Finan. Math. 1(1), 289--325.
\newblock ISSN 1945-497X.
\newblock \doi{10.1137/090753814}.

\bibitem[{Cox, Ingersoll, and Ross(1985)}]{CIR85}
\textsc{Cox, J.~C., Ingersoll, J.~E., and Ross, S.~A.} (1985).
\newblock \emph{A theory of the term structure of interest rates}.
\newblock Econometrica 53(2), 385--407.
\newblock ISSN 0012-9682.
\newblock \doi{10.2307/1911242}.

\bibitem[{Duffie, Pan, and Singleton(2000)}]{Duffie00}
\textsc{Duffie, D., Pan, J., and Singleton, K.} (2000).
\newblock \emph{Transform analysis and asset pricing for affine
  jump-diffusions}.
\newblock Econometrica 68(6), 1343--1376.
\newblock ISSN 0012-9682.
\newblock \doi{10.1111/1468-0262.00164}.

\bibitem[{Elices(2008)}]{Elices08}
\textsc{Elices, A.} (2008).
\newblock \emph{{Models with time-dependent parameters using transform methods:
  application to Heston's model}}.
\newblock Available at arXiv: \url{https://arxiv.org/abs/0708.2020}.

\bibitem[{Fouque, Papanicolaou, and Sircar(2000)}]{Fouque00}
\textsc{Fouque, J.-P., Papanicolaou, G., and Sircar, K.~R.} (2000).
\newblock Derivatives in financial markets with stochastic volatility.
\newblock Cambridge University Press, Cambridge, U.K.
\newblock ISBN 0-521-79163-4.

\bibitem[{Fouque, Papanicolaou, Sircar, and Solna(2003)}]{Fouque03a}
\textsc{Fouque, J.-P., Papanicolaou, G., Sircar, R., and Solna, K.} (2003).
\newblock \emph{Multiscale stochastic volatility asymptotics}.
\newblock Multiscale Model. Simul. 2(1), 22--42.
\newblock ISSN 1540-3459.
\newblock \doi{10.1137/030600291}.

\bibitem[{Gatheral(2006)}]{Gatheral06}
\textsc{Gatheral, J.} (2006).
\newblock The volatility surface: A practitioner's guide.
\newblock Wiley Finance. John Wiley \& Sons, Hoboken, New Jersey.
\newblock ISBN 9780470068250.

\bibitem[{Gulisashvili and Vives(2012)}]{GulisashviliVives12}
\textsc{Gulisashvili, A. and Vives, J.} (2012).
\newblock \emph{Two-sided estimates for distribution densities in models with
  jumps}.
\newblock In Stochastic Differential Equations and Processes, pp. 239--254.
  Springer, Berlin, Heidelberg.
\newblock \doi{10.1007/978-3-642-22368-6_7}.

\bibitem[{Hanson(2007)}]{Hanson07}
\textsc{Hanson, F.~B.} (2007).
\newblock Applied stochastic processes and control for jump-diffusions, vol.~13
  of \emph{Advances in Design and Control}.
\newblock SIAM, Philadelphia, PA.
\newblock ISBN 9780898716337.

\bibitem[{Heston(1993)}]{Heston93}
\textsc{Heston, S.~L.} (1993).
\newblock \emph{A closed-form solution for options with stochastic volatility
  with applications to bond and currency options}.
\newblock Rev. Financ. Stud. 6(2), 327--343.
\newblock ISSN 0893-9454.
\newblock \doi{10.1093/rfs/6.2.327}.

\bibitem[{Hull and White(1987)}]{HullWhite87}
\textsc{Hull, J.~C. and White, A.~D.} (1987).
\newblock \emph{The pricing of options on assets with stochastic volatilities}.
\newblock J. Finance 42(2), 281--300.
\newblock ISSN 1540-6261.
\newblock \doi{10.1111/j.1540-6261.1987.tb02568.x}.

\bibitem[{Jafari and Vives(2013)}]{JafariVives13}
\textsc{Jafari, H. and Vives, J.} (2013).
\newblock \emph{A {H}ull and {W}hite formula for a stochastic volatility
  {L}\'evy model with infinite activity}.
\newblock Commun. Stoch. Anal. 7(2), 321--336.
\newblock ISSN 0973-9599.

\bibitem[{Killmann and von Collani(2001)}]{Killmann01}
\textsc{Killmann, F. and von Collani, E.} (2001).
\newblock \emph{A note on the convolution of the uniform and related
  distributions and their use in quality control}.
\newblock Econ. Qual. Control 16(1), 17--41.
\newblock ISSN 0940-5151.
\newblock \doi{10.1515/EQC.2001.17}.

\bibitem[{Kou(2002)}]{Kou02}
\textsc{Kou, S.~G.} (2002).
\newblock \emph{A jump-diffusion model for option pricing}.
\newblock Manage. Sci. 48(8), 1086--1101.
\newblock ISSN 0025-1909.
\newblock \doi{10.1287/mnsc.48.8.1086.166}.

\bibitem[{Lewis(2000)}]{Lewis00}
\textsc{Lewis, A.~L.} (2000).
\newblock Option Valuation Under Stochastic Volatility: With {M}athematica
  code.
\newblock Finance Press, Newport Beach, CA.
\newblock ISBN 9780967637204.

\bibitem[{Merino and Vives(2015)}]{MerinoVives15}
\textsc{Merino, R. and Vives, J.} (2015).
\newblock \emph{A generic decomposition formula for pricing vanilla options
  under stochastic volatility models}.
\newblock Int. J. Stoch. Anal. pp. Art. ID 103647, 11.
\newblock ISSN 2090-3332.
\newblock \doi{10.1155/2015/103647}.

\bibitem[{Merino and Vives(2017)}]{MerinoVives17}
\textsc{Merino, R. and Vives, J.} (2017).
\newblock \emph{Option price decomposition in spot-dependent volatility models
  and some applications}.
\newblock Int. J. Stoch. Anal. pp. Art. ID 8019498, 16.
\newblock ISSN 2090-3332.
\newblock \doi{10.1155/2017/8019498}.

\bibitem[{Merton(1976)}]{Merton76}
\textsc{Merton, R.~C.} (1976).
\newblock \emph{Option pricing when underlying stock returns are
  discontinuous}.
\newblock J. Financ. Econ. 3(1--2), 125--144.
\newblock ISSN 0304-405X.
\newblock \doi{10.1016/0304-405X(76)90022-2}.

\bibitem[{Mikhailov and N{\"o}gel(2003)}]{MikhailovNogel03}
\textsc{Mikhailov, S. and N{\"o}gel, U.} (2003).
\newblock \emph{{H}eston's stochastic volatility model - implementation,
  calibration and some extensions.}
\newblock Wilmott magazine 2003(July), 74--79.

\bibitem[{Mr{\'a}zek, Posp\'{\i}{\v{s}}il, and
  Sobotka(2016)}]{MrazekPospisilSobotka16ejor}
\textsc{Mr{\'a}zek, M., Posp\'{\i}{\v{s}}il, J., and Sobotka, T.} (2016).
\newblock \emph{On calibration of stochastic and fractional stochastic
  volatility models}.
\newblock European J. Oper. Res. 254(3), 1036--1046.
\newblock ISSN 0377-2217.
\newblock \doi{10.1016/j.ejor.2016.04.033}.

\bibitem[{Posp\'{\i}\v{s}il and Sobotka(2016)}]{PospisilSobotka16amf}
\textsc{Posp\'{\i}\v{s}il, J. and Sobotka, T.} (2016).
\newblock \emph{Market calibration under a long memory stochastic volatility
  model}.
\newblock Appl. Math. Finance 23(5), 323--343.
\newblock ISSN 1350-486X.
\newblock \doi{10.1080/1350486X.2017.1279977}.

\bibitem[{Scott(1987)}]{Scott87}
\textsc{Scott, L.~O.} (1987).
\newblock \emph{Option pricing when the variance changes randomly: Theory,
  estimation, and an application}.
\newblock J. Financ. Quant. Anal. 22(4), 419--438.
\newblock ISSN 0022-1090.
\newblock \doi{10.2307/2330793}.

\bibitem[{Scott(1997)}]{Scott97}
\textsc{Scott, L.~O.} (1997).
\newblock \emph{Pricing stock options in a jump-diffusion model with stochastic
  volatility and interest rates: Applications of fourier inversion methods}.
\newblock Math. Finance 7(4), 413--426.
\newblock ISSN 0960-1627.
\newblock \doi{10.1111/1467-9965.00039}.

\bibitem[{Stein and Stein(1991)}]{SteinStein91}
\textsc{Stein, J. and Stein, E.} (1991).
\newblock \emph{Stock price distributions with stochastic volatility: An
  analytic approach}.
\newblock Rev. Financ. Stud. 4(4), 727--752.
\newblock ISSN 0893-9454.
\newblock \doi{10.1093/rfs/4.4.727}.

\bibitem[{Vives(2016)}]{Vives16}
\textsc{Vives, J.} (2016).
\newblock \emph{Decomposition of the pricing formula for stochastic volatility
  models based on malliavin-skorohod type calculus}.
\newblock In \textsc{M.~Eddahbi, E.~H. Essaky, and J.~Vives}, eds., Statistical
  Methods and Applications in Insurance and Finance: CIMPA School, Marrakech
  and El Kelaa M'gouna, Morocco, April 2013, pp. 103--123. Springer, Cham.
\newblock ISBN 978-3-319-30417-5.
\newblock \doi{10.1007/978-3-319-30417-5_4}.

\bibitem[{Yan and Hanson(2006)}]{YanHanson06}
\textsc{Yan, G. and Hanson, F.~B.} (2006).
\newblock \emph{Option pricing for a stochastic-volatility jump-diffusion model
  with log-uniform jump-amplitude}.
\newblock In Proceedings of American Control Conference, pp. 2989--2994. IEEE,
  Piscataway, NJ.
\newblock \doi{10.1109/acc.2006.1657175}.

\end{thebibliography}
\end{document}